\documentclass[a4paper]{amsart}


\usepackage{amssymb,latexsym}
\usepackage{amsmath,amsthm}
\usepackage{graphicx}
\usepackage{caption}
\usepackage{subcaption}
\usepackage[usenames,dvipsnames]{color}
\usepackage{float}
\usepackage{tikz,tikz-3dplot}
\usepackage[backref=page,colorlinks=true,linkcolor=Orchid,citecolor=Orchid,urlcolor=Orchid]{hyperref}
\usepackage{marginnote}
\usepackage{enumitem}

\newtheorem{theorem}{Theorem}
\newtheorem{lemma}{Lemma}
\newtheorem{definition}{Definition}

\newtheorem*{perfect fluid}{Vacuum and perfect fluid cases}

\newcommand{\R}{\mathbb{R}}

\begin{document}


\title{Kantowski-Sachs cosmology with Vlasov matter}

\author{\textsc{David FAJMAN}}
\author{\textsc{Gernot HEI}\ss\textsc{EL}}
\thanks{The authors acknowledge support of the Austrian Science Fund (FWF) through the Project \emph{Geometric transport equations and the non-vacuum Einstein flow} (P 29900-N27).}

\address{\href{http://gravity.univie.ac.at}{Gravitational Physics Group\\
	Faculty of Physics\\
	University of Vienna\\
	Austria}}
	
	

\begin{abstract}
We analyse the Kantowski-Sachs cosmologies with Vlasov matter of massive and massless particles using dynamical systems analysis. We show that generic solutions are past and future asymptotic to the non-flat locally rotationally symmetric Kasner vacuum solution. Furthermore, we establish that solutions with massive Vlasov matter behave like solutions with massless Vlasov matter towards the singularities.
\end{abstract}

\maketitle



\section{Introduction}\label{S: introduction}

The Einstein-Vlasov system describes spacetimes containing ensembles of self-gravitating, collisionless particles and constitutes an excellent model for the large-scale structure of the universe where collisions of particles (modeling galaxies or galaxy clusters) are indeed rare. For the class of cosmological spacetimes certain stability results were resolved in recent years \cite{AnderssonFajman2017, Ringstroem2013} (cf \cite{FajmanEtAl2017, LindbladTaylor2017, Taylor2017} for recent stability results for the Einstein-Vlasov system in the asymptotically flat case).

Understanding cosmological dynamics of Einstein-matter systems is in general based on the study of dynamics of homogeneous cosmological models, whose
stability properties are analyzed in a next step. The program of classifying dynamics of homogeneous spacetimes with Vlasov matter contains a number of
results but for certain classes this problem is still open, among those are the Kantowski-Sachs models, which are the subject of the present paper.

\subsection{The dynamical systems approach to spatially homogenous Einstein-Vlasov cosmology}\label{SS: DS approach}

Spatially homogenous (SH) cosmologies with non-tilted perfect fluids with linear equations of state have been analysed with great success using dynamical systems methods; cf~\cite{Coley2003, HorwoodEtAl2003, WainwrightEllis1997}. This is due to the fact that in this context, the Einstein equations reduce to a system of ODEs. In contrast, the Einstein-Vlasov system remains a system of PDEs even in the spatially homogenous context; cf~\cite{Andreasson2011, Rendall2004, Rendall2008}.

Despite that, a route to analyse SH Einstein-Vlasov cosmologies with dynamical systems methods has been initiated by Rendall~\cite{Rendall1996}. The core idea is to work with distribution functions which exhibit the same symmetries as the underlying spacetime in a manner which we make precise in Section~\ref{S: SH LRS f} (Definition~\ref{D: 1}), and which then satisfy the Vlasov equation. The general functional form of such distributions has been investigated by Maartens and Maharaj in \cite{MaartensMaharaj1990} in the context of spatial homogeneity, and in~\cite{MaartensMaharaj1985} for SH cosmologies which are also locally rotationally symmetric (LRS). The cosmologies for which they succeeded without additional assumptions, and for which the found distributions are non-trivial are:
\begin{align}
&\text{SH but not LRS:}	&	&\text{Bianchi I.}	\label{E: SH list} \\
&\text{SH and LRS:}		&	&\text{LRS Bianchi I, II, III, VII$_0$, VIII, IX, Kantowski-Sachs.}	\label{E: SH LRS list}
\end{align}
The latter are all SH cosmologies which are compatible with local rotational symmetry, except Bianchi~V and VII$_h$.\footnote{Note that Bianchi type~III is the same as Bianchi type~VI$_{-1}$, cf~\cite[p~37]{WainwrightEllis1997}, and that LRS Bianchi types~I and VII$_0$ can be identified with each other; cf~\cite[p~2579--80, p~2584]{WainwrightEtAl1999}, \cite[p~1715]{RendallTod1999} or~\cite[App~B, p~667]{CalogeroHeinzle2011}.}

Though restricting to these distribution functions comes with the price of loosing generality it has the advantage that the Vlasov equation is solved explicitly and thus eliminated from the system. Hence what is left to solve is the Einstein part, which then again reduces to an autonomous system of ODEs. As with perfect fluids, the Vlasov matter then enters the dynamics only by the energy-momentum tensor. However, since the latter is given by integrals over the distribution function which is arbitrary up to its general functional form, Vlasov matter is still more difficult to deal with than perfect fluids, for which the energy momentum tensor is fixed by a single parameter. The approach taken by Rendall~\cite{Rendall1996} is to make the matter integrals tractable by imposing a further symmetry on the distribution functions -- reflection symmetry, which we will define in Section~\ref{S: SH LRS f} (Definition~\ref{D: 2}).

\subsection{Alternatives}

The strength of the dynamical systems approach to SH Einstein-Vlasov cosmology described above is that it allows for a global analysis of the state-space. The downside is that one has to impose certain symmetry assumptions in order to obtain a tractable dynamical system, thereby loosing generality. We thus refer to Nungesser et al.~\cite{LeeNungesser2017, LeeNungesser2018, Nungesser2010, Nungesser2012, Nungesser2013, NungesserEtAl2013} for an alternative route -- a small data future stability analysis within the class of SH spacetimes. The strength of this approach is that the Vlasov distributions are of full generality. The downside is that it does not yield global results on the full state-space.

\subsection{Literature}\label{SS: literature}

SH spacetimes can be classified by the properties of the underlying symmetry group. One distinguishes between Bianchi and Kantowski-Sachs spacetimes; cf~\cite[Subsec~1.2.2]{WainwrightEllis1997} and Subsection~\ref{SS: sym grp} below. The class of Bianchi spacetimes is divided further into types~I--IX. Bianchi~I has been analysed in the dynamical systems approach described in Subsection~\ref{SS: DS approach} in~\cite{Rendall1996} and~\cite{HeinzleUggla2006}. LRS Bianchi~I, II and III are covered by~\cite{Rendall2002, RendallTod1999, RendallUggla2000} and LRS Bianchi~VIII and~IX by~\cite{CalogeroHeinzle2010, Heissel2012}. Kantowski-Sachs has been investigated in~\cite{RendallTod1999}, however with the restriction to Vlasov matter of massless type. This leaves the corresponding Kantowski-Sachs analysis with massive Vlasov matter as the only open problem from the list \eqref{E: SH list}, \eqref{E: SH LRS list}. The present paper closes this gap.

There is also a series of papers \cite{CalogeroHeinzle2009, CalogeroHeinzle2011, Heissel2012} initiated by Calogero and Heinzle which is closely related. These are concerned with a broadly defined, generally anisotropic, matter family. The reason why this is of interest here is threefold. Firstly, as pointed out in \cite[Sec~7.1]{CalogeroHeinzle2009} and~\cite[Sec~12.1]{CalogeroHeinzle2011}, massless Vlasov matter subject to the same symmetry assumptions which we impose here naturally fits into their matter family as a special case. Secondly, in~\cite{CalogeroHeinzle2010, Heissel2012} the formalism could be extended to incorporate massive Vlasov matter as well. Thirdly, \cite{CalogeroHeinzle2011} readily lists the dynamical system for Kantowski-Sachs cosmologies in their original anisotropic matter framework. We are adopt this system here.

\subsection{Results and outline}

In this paper we give a dynamical system analysis of Kantowski-Sachs cosmologies with massive or massless Vlasov, thereby closing the literature gap identified in subsection~\ref{SS: literature}. The distribution functions considered are invariant under the symmetries of the spacetime (Definition~\ref{D: 1}), reflection symmetric (Definition~\ref{D: 2}) as well as split supported (Definition~\ref{D: 3}). Our main result (Theorem~\ref{T: 1}) is that generic solutions are past and future asymptotic to the non-flat LRS Kasner vacuum solution, and that there are non-generic solutions which are asymptotic to a non-isotropic Bianchi~I matter solution or the flat Friedman matter solution. We give the metric to all these solutions. For the case of massless particles, we thereby recover~\cite[Thm~5.1]{RendallTod1999}. The result for massive particles is new. Our results are in accordance with results on recollapse of the spacetime, cf Subsection~\ref{SS: recollapse}, we also confirm that all solutions recollapse. Finally, we establish that solutions with massive particles approach solutions with massless particles towards the singularities.

In Section~\ref{S: KS cosmologies} we give some background on Kantowski-Sachs cosmologies. Section~\ref{S: SH LRS f} outlines the assumptions imposed on the distribution functions. The dynamical system for Kantowski-Sachs cosmologies with Vlasov matter is then formulated in Section~\ref{S: DS for KS}, while Section~\ref{S: state-space} deals with the corresponding state-space. Finally, we present our results in Section~\ref{S: results}. Appendix~\ref{A: flow at infinity} is concerned with the analysis of the flow at infinity.

We will assume some familiarity with dynamical systems theory. For a background we refer to~\cite{Perko2001} or~\cite[Chap~4]{WainwrightEllis1997}. Greek indices are space-time indices, while latin indices are spatial. We use the Einstein summation convention and sum over repeated indices if not indicated otherwise.

\section{Kantowski-Sachs cosmologies}\label{S: KS cosmologies}

In this section we give some background on the class of SH cosmologies called Kantowski-Sachs cosmologies. In Subsection~\ref{SS: sym grp} we give the definition in terms of the underlying symmetry group, and discuss the structure of the associated Lie-algebra. In Subsection~\ref{SS: geom & top} we give the metric in a symmetry adapted frame, and list the topologies on which it can be realised. Finally, in Subsection~\ref{SS: recollapse} we review some results concerning the recollapse property of this class, and comment on how our results relate to those.

\subsection{Symmetry group}\label{SS: sym grp}
By a \emph{cosmology} we refer to a spacetime that solves Einstein's equations.
\begin{definition}\label{D: KS}
Kantowski-Sachs cosmologies are the class of cosmologies with a four-dimensional continuous Lie-group of isometries which acts multiply-transitively on three-dimensional spatial hypersurfaces, but which does not exhibit a three-dimensional subgroup which acts simply-transitively on them; cf~\cite{Collins1977} or~\cite[Sec~1.2.2]{WainwrightEllis1997}.
\end{definition}
By the former statement of Definition~\ref{D: KS}, Kantowski-Sachs cosmologies are SH and LRS. The latter statement of Definition~\ref{D: KS} is what distinguishes Kantowski-Sachs cosmologies from SH and LRS Bianchi cosmologies, for which such a subgroup does exist.\footnote{In contrast to Kantowski-Sachs, Bianchi cosmologies are not necessarily LRS.} The Kantowski-Sachs cosmologies however do admit a three-dimensional subgroup of isometries, which acts simply transitively on two-spheres. It can hence be identified with $SO(3,\R)$, and we denote the associated generators, ie Killing vector fields, by $\{\xi_1,\xi_2,\xi_3\}$. In addition to this spherical symmetry, Kantowski-Sachs cosmologies also exhibit a (radial) translation symmetry. Denoting the corresponding Killing vector field by $\eta$, the Lie-algebra associated with the full isometry group can be represented by $[\xi_1,\xi_2]=\xi_3$ and cyclic permutations thereof, together with $[\eta,\xi_i]=0$ $\forall\, i\in\{1,2,3\}$; cf~\cite[App~B]{Collins1977} or~\cite[App~B.2]{CalogeroHeinzle2011}.\footnote{The extended Schwarzschild spacetime also admits this symmetry inside the event-horizon -- outside the Lie-algebra is the same, but $\eta$ is timelike; cf~\cite[Sec~1]{Collins1977}.}

\subsection{Geometry and topology}\label{SS: geom & top}

A convenient choice of coframe for Kantowski-Sachs is given by a time independent covector $\hat\omega^1$ which is invariant under the radial translation symmetry, together with an arbitrary time independent orthonormal frame $\{\hat\omega^2,\hat\omega^3\}$ on the two-spheres; cf~\cite[Sec~2]{RendallTod1999}. The hat emphasises time independence. Choosing an appropriate time coordinate $t$, a general Kantowski-Sachs metric then takes the form
\begin{equation}\label{E: SH LRS metric}
^4g=-\,\mathrm dt\otimes\mathrm dt+g_{11}(t)\,\hat\omega^1\otimes\hat\omega^1+g_{22}(t)\,(\hat\omega^2\otimes\hat\omega^2+\hat\omega^3\otimes\hat\omega^3) \, .
\end{equation}
This choice is canonical in the sense that spatial homogeneity is reflected by the spatial metric depending on $t$ only, and we choose the spatial frame such, that the time dependence rests solely on the spatial metric components, while the 1-forms $\hat\omega^i$ are time-independent. Local rotational symmetry on the other hand is reflected by the spatial metric components satisfying $g_{22}(t)=g_{33}(t)$.

The spatial part of the metric~\eqref{E: SH LRS metric} can be realised on topologies of the form $\R\times S^2$, or topologies which can be derived from this by (\emph{i}) identifying points under a translation in $\eta$-direction, ie in the $\R$-part of the topology, or a translation in $\eta$-direction together with a rotation or reflection, (\emph{ii}) an identification of antipodal points in each $S^2$ or (\emph{iii}) a combination of (\emph{i}) and (\emph{ii}); cf~\cite[Sec~2]{Collins1977}. A prominent example is $S\times S^2$. Different realisations of these topologies do not effect the analysis we perform. In particular our results hold for all of those topologies.

\subsection{Recollapse}\label{SS: recollapse}

There are several results concerning recollapse, ie the presence of both, a big-bang singularity and a big-crunch singularity. \cite[Thm~2]{Collins1977} shows recollapse by geodesic incompleteness for Kantowski-Sachs cosmologies with perfect fluids. \cite[Thm~1.2]{Burnett1991} proofs recollapse by showing that the length of timelike curves in Kantowski-Sachs cosmologies is bounded from above, given that the stress tensor of the matter is positive definite, ie that it satisfies the non-negative-sum-pressures condition. This result is embedded in the more general result \cite[Thm~1.1]{Burnett1991}, which shows that there is an upper bound to the length of timelike curves in spherically symmetric spacetimes that possess $S\times S^2$ Cauchy surfaces and that satisfy the non-negative-pressure and dominant energy conditions. It thereby also includes Kantowski-Sachs cosmologies with Vlasov matter. A related recollapse result concerning Einstein-Vlasov spacetimes with spherical symmetry, but in general without additional translation symmetry, is also given in~\cite[Thm~3.11]{Henkel2002}. However, none of these results establishes the precise behaviour towards the singularities which we obtain.

\section{Spatially homogenous, locally rotationally symmetric Vlasov distributions}\label{S: SH LRS f}

The purpose of this section is to define Vlasov distributions exhibiting the symmetries of the underlying spacetime, which is what we assume in our analysis. In Subsection~\ref{SS: EVS} we briefly introduce the Einstein-Vlasov system. Subsection~\ref{SS: inheritance of sym} discusses the way in which distribution functions inherit the symmetries of the underlying spacetime. In Subsection~\ref{SS: sym inv f} we give the definition of symmetry invariant distribution functions, and discuss the resulting general functional forms of those functions int he context of SH as well as SH LRS spacetimes. Finally, in Subsection~\ref{SS: further sym} we discuss further assumptions imposed on the distribution functions.

\subsection{The Einstein-Vlasov system}\label{SS: EVS}

Consider the mass-shell
\begin{align}\notag
P_m:=\{ (x^\lambda,v^\lambda) \,|\, x^\lambda\in M , v^\lambda\in T_x M, v_\mu v^\mu=-m^2, v^0>0 \}
\end{align}
associated with a spacetime $(M,{^4g})$. Let there be an ensemble of collisionless particles of mass $m=1$ or $0$, described by the distribution function $f(x^\lambda,v^\lambda)$, where $x^\lambda$ and $v^\lambda$ stand for the particle’s positions and four-momenta with domain $P_m$. In the case $m=0$ we understand $v^i=0$ to be excluded from the domain. Since the particles are collisionless, $f$ satisfies the Vlasov equation, which in coordinates $\{t,x^i\}$ and with Christoffel symbols $\Gamma^\rho_{\mu\nu}$ reads
\begin{align}\label{E: Vlasov}
v^\mu\frac{\partial f}{\partial x^\mu} - \Gamma^k_{\mu\nu}v^\mu v^\nu\frac{\partial f}{\partial v^k} = 0 .
\end{align}
The Einstein-Vlasov system is then given by~\eqref{E: Vlasov} together with the Einstein equations $R_{\mu\nu}-\frac{1}{2}g_{\mu\nu}R=T_{\mu\nu}$, with an energy-momentum tensor of the form
\begin{align}\label{E: T Vlasov}
T^{\mu\nu} = \int f(x^\lambda,v^\lambda)\,v^\mu v^\nu\,\mathrm{vol} ,\quad\text{where}\quad
\mathrm{vol} := \frac{\sqrt{|\det{^4g}|}}{|v_0|} \, \mathrm dv^1\mathrm dv^2 \mathrm dv^3
\end{align}
denotes the volume form induced by $^4g$ on $P_m$. Hence, in~\eqref{E: T Vlasov}, $|v_0|$ is understood to be constrained by the relation
\begin{align}\notag
{^4g}(v,v) =
\begin{cases}
-1, 			& m=1 \\
\phantom{-}0,	& m=0
\end{cases} .
\end{align}
One obtains directly from the definitions~\eqref{E: Vlasov} and~\eqref{E: T Vlasov} that Vlasov matter satisfies the dominant and strong energy conditions, as well as the non-negative pressures condition; cf~\cite[Sec~1]{Rendall2004}. For more details on the Einstein-Vlasov system we refer to~\cite{Andreasson2011, Rendall2004, Rendall2008}.

\subsection{Inheritance of symmetries by the distribution function}\label{SS: inheritance of sym}

The dynamical systems approach to spatially homogenous Einstein-Vlasov cosmology, which we will use in the following, relies on symmetries obeyed by the distribution function~$f$, which, by virtue of these symmetries, then solves the Vlasov equation~\eqref{E: Vlasov} and effectively eliminates the latter from the system of equations. In the proceeding discussion, we mainly follow the work by Maartens and Maharaj \cite{MaartensMaharaj1985,MaartensMaharaj1990}.

Let the spacetime $(M,{^4g},f)$ be subject to symmetries, described by a set of Killing vector fields $X$, then
\begin{align}\label{E: Killing}
\mathcal L_\xi{^4g}=0\quad \forall\,\xi\in X \,,
\end{align}
where $\mathcal L$ denotes the Lie-derivative. In the case of Kantowski-Sachs spacetimes there are four such fields -- three corresponding to spatial homogeneity, and one to local rotational symmetry; cf \cite[Sec~1.2.2 (b)]{WainwrightEllis1997}. Through the Einstein equations, the symmetries~\eqref{E: Killing} are inherited by the energy-momentum tensor, such that\footnote{This relation also holds for homothetic vector fields, for which $\mathcal L_\xi{^4g}=\kappa$, with $\kappa=\mathrm{const}$, since the Riemann tensor is invariant with respect to a constant scaling; cf~\cite[Sec~1.2.4]{WainwrightEllis1997}.}
\begin{align}\label{E: T sym}
\mathcal L_\xi T_{\mu\nu}=0 \quad \forall\,\xi\in X \,.
\end{align}
Consequently, from~\eqref{E: T sym} with~\eqref{E: T Vlasov}, the distribution function $f$ inherits the symmetry by
\begin{align}\label{E: f sym}
\mathcal L_\xi T_{\mu\nu} = \int \tilde\xi(f)\,v^\mu v^\nu \,\mathrm{vol} = 0 \quad\forall\,\xi\in X ,
\end{align}
where
\begin{align}
\tilde\xi = \xi^i\frac{\partial}{\partial x^i}+{\xi^i}_{,j}v^j\frac{\partial}{\partial v^i}
\end{align}
denotes the complete lift of the Killing vector $\xi$ onto the tangent bundle; cf~\cite[Sec~II.C]{MaartensMaharaj1985}, \cite[Sec~III.E]{BerezdivinSachs1973}, references therein and~\cite{SarbachZannias2014}. Hence, a Vlasov matter distribution in a space-time with symmetries $X$ has to satisfy~\eqref{E: f sym}.

\subsection{Symmetry invariant distribution functions}\label{SS: sym inv f}

\begin{definition}\label{D: 1}
A Vlasov distribution functions $f$ is invariant under the symmetries of the spacetime described by a set of Killing vector fields $X$ if and only if
\begin{align}\label{E: f condition}
\tilde\xi(f) = 0 \quad \forall\,\xi\in X \,.
\end{align}
An $f$ which is invariant under spatial homogeneity (and local rotational symmetry) is called a spatially homogenous (and locally rotationally symmetric) distribution function.
\end{definition}
Clearly,~\eqref{E: f sym} follows from~\eqref{E: f condition}, but the reverse statement does not hold in general.\footnote{A counter example is presented in~\cite{EllisEtAl1983i} and also quoted in~\cite[Sec~C]{MaartensMaharaj1985}: a $k=0$ Friedman spacetime on which there lives an anisotropic $f$.} Hence, considering such distribution functions usually mean a loss of generality. In~\cite{MaartensMaharaj1985,MaartensMaharaj1990}, Maartens and Maharaj sought to find the general functional forms of SH as well as SH and LRS distribution functions which satisfy the Vlasov equation~\eqref{E: Vlasov}. The cosmologies for which they succeeded without additional assumptions, and for which the found distributions are non-trivial are~\eqref{E: SH list}, \eqref{E: SH LRS list}. Those distribution functions are:
\begin{align}
f&=f_0(v_1,v_2,v_3),			& &\text{Bianchi I;} \label{E: f0} \\
f&=f_0\big(v_1,v_2^2+v_3^2\big),	& &\text{SH LRS cosmologies listed in~\eqref{E: SH LRS list}.} \label{E: f0 LRS}
\end{align}
Here $v_i$ denote the Killing vector constants of motion corresponding to spatial homogeneity, ie the conserved momenta of the particles. Spatial homogeneity is reflected in~\eqref{E: f0}, \eqref{E: f0 LRS} by the fact that $f_0$ is independent of the spatial coordinates. In addition, local rotational symmetry is reflected in~\eqref{E: f0 LRS} by the fact that $f_0$ is independent of the direction of the particle momenta in the $2,3$-plane. Note also that the $f_0$ are time independent.

\subsection{Further symmetry assumptions}\label{SS: further sym}

The route initiated by Rendall~\cite{Rendall1996} is based on choosing Vlasov distribution functions of the form~\eqref{E: f0}, \eqref{E: f0 LRS} and to make the matter integrals tractable by imposing a further symmetry on the distribution functions -- \emph{reflection symmetry}. We will restrict here to the SH LRS case, and refer to~\cite{Rendall1996} for the analogous discussion in the context of Bianchi~I without additional LRS symmetry.
\begin{definition}\label{D: 2}
A non-trivial spatially homogenous and locally rotationally symmetric Vlasov distribution function $f_0\big(v_1,v_2^2+v_3^2\big)$ is reflection symmetric if and only if it is even in $v_1$.
\end{definition}
From~\eqref{E: T Vlasov} it then follows that the corresponding energy-momentum tensor is diagonal, and so is the metric for SH LRS spacetimes; cf~\cite{MaartensMaharaj1985}. Approaches using reflection symmetry are thus also referred to as diagonal models. We refer to~\cite{RendallTod1999, RendallUggla2000} for more details.

One further restriction sometimes imposed, and in this paper as well, is \emph{split support}; cf~\cite{RendallUggla2000}.
\begin{definition}\label{D: 3}
A spatially homogenous and locally rotationally symmetric distribution function $f_0\big(v_1,v_2^2+v_3^2\big)$ has split support if and only if its support does not intersect the coordinate planes $v^i=0$.
\end{definition}
The motivation for this is to ensure that certain matter functions are smooth, as we will see in Subsection~\ref{SS: matter functions}

\section{The dynamical system for Kantowski-Sachs cosmologies with Vlasov matter}\label{S: DS for KS}

In this section we prepare the equations for our analysis. The dynamical system for Kantowski-Sachs cosmologies with Vlasov matter is given in Subsection~\ref{SS: dynamical system}. The Vlasov matter parameters which enter the system are discussed in Subsection~\ref{SS: matter functions}.

\subsection{The dynamical system}\label{SS: dynamical system}

Using their anisotropic matter family, in~\cite[Sec~9]{CalogeroHeinzle2011} Calogero and Heinzle present the corresponding three-dimensional dynamical systems for LRS Bianchi~VIII, IX and III as well as for Kantowski-Sachs in a unified form. As stated in Subsection~\ref{SS: literature}, this matter family also encompasses the case of massless Vlasov matter, subject to the symmetry conditions which we impose here -- Definitions~\ref{D: 1}, \ref{D: 2} and~\ref{D: 3}. In the case of LRS Bianchi~IX, the extension of the system to also include the case of massive Vlasov matter could be achieved by dropping the linearity of the equation of state of the matter on the one hand, cf Subsection~\ref{SS: matter functions}, and introducing an additional variable and evolution equation to the system on the other hand; cf~\cite[Sec~3]{CalogeroHeinzle2010}. Due to the unified framework, the same generalisation applies to the case of LRS Bianchi~VIII, cf~\cite{Heissel2012}, and it is straight forward to see that this is also the case for Kantowski-Sachs.

Hence, from~\cite[Sec~9]{CalogeroHeinzle2011} and~\cite[Sec~3]{CalogeroHeinzle2010}, we can readily write down the four-dimensional Einstein-Vlasov dynamical system for Kantowski-Sachs cosmologies: It consists of the evolution equations
\begin{align} 
H_D' &= -(1-H_D^2)(q-H_D\Sigma_+) \, , \label{E: HDprime} \\
\Sigma_+' &= -(2-q)H_D\Sigma_+-(1-H_D^2)(1-\Sigma_+^2)+\Omega(w_2(l,s)-w_1(l,s)) \, , \label{E: Sigmaprime} \\
M_1' &= M_1(qH_D-4\Sigma_++(1-H_D^2)\Sigma_+) , \label{E: Mprime} \\
l' &= 2H_D l (1-l) . \label{E: lprime}
\end{align}
together with the Hamiltonian constraint
\begin{align}\label{E: hc}
\Omega &= 1-\Sigma_+^2 \, .
\end{align}

The four dynamical quantities are defined by
\begin{align}
H_D&:=-\frac{1}{3D}({k^1}_1+2{k^2}_2) \, ,   &   \Sigma_+&:=\frac{1}{3D}({k^1}_1-{k^2}_2) \, , \\
M_1&:=\frac{1}{D}\frac{\sqrt{g_{11}}}{g_{22}} \, ,         &  l&:=\frac{(\det g)^{1/3}}{1+(\det g)^{1/3}} \, , \label{E: M1 and l}
\end{align}
where ${k^i}_j$ are the components of the extrinsic curvature. $H:=DH_D$ is the Hubble scalar, ie the expansion scalar; cf~\cite[Sec~1.1.3]{WainwrightEllis1997}. The dominant variable $D$ is given by
\begin{align}\label{E: D}
D:=\sqrt{H^2+\frac{1}{3g_{22}}}>0 \, ;
\end{align}
cf also~\cite[Sec~8.5.2]{WainwrightEllis1997}. $D\Sigma_+$ is the only independent degree of freedom of the components of the shear tensor; cf~\cite[Sec~1.1.3]{WainwrightEllis1997}. $DM_1$ gives the ratio between the spatial metric components belonging to the plane of local rotational symmetry, and the (square-root of the) spatial metric component in the orthogonal direction. $l$ is a measure of a length scale via the third root of the spatial volume element $\det g$, however compactified to the range $(0,1)$. All four dynamical quantities are dimensionless. The prime denotes derivation with respect to $D$-rescaled time
\begin{align}
\tau := \int_{t_0}^t D(t)\,\mathrm dt \quad\text{with}\quad
\tau(t_0)=0 .
\end{align}

Next we specify the quantities $\Omega$, $q$ and $s$. $\Omega$ is a dimensionless measure of the energy density of the matter. More precisely, if $\rho:=T_{tt}$ denotes the energy density, where $T_{ij}$ are the components of the energy-momentum tensor, then
\begin{align}\notag
\Omega&:=\frac{\rho}{3D^2} .	&	q&:=2\Sigma_+^2+\frac{1}{2}(1+3w(l,s))\Omega
\end{align}
denotes the deceleration parameter, and it is a function of $\Sigma_+$ and the matter parameter $w(l,s)$, which we will specified in Subsection~\ref{SS: matter functions} below. Finally, $s$ is a measure of the anisotropy of the inverse metric components $g^{ii}$, which can be expressed as function of $H_D$ and $M_1$:
\begin{align}\label{E: s}
s &:= \frac{g^{22}}{\sum_i g^{ii}} = \left( 2 + \frac{3(1-H_D^2)}{M_1^2}  \right)^{-1} \in(0,1/2).
\end{align}
In Section~\ref{S: state-space} we will see that we can also think of $s$ as a rotational angle around the $H_D = \pm1$, $M_1=0$ axes in the state-space; cf Figure~\ref{F: state-space X0}.

The Vlasov matter functions $w_1(l,s)$, $w_2(l,s)$ and $w(l,s)$ are functions of $l$ and $s$, which we give in Subsection~\ref{SS: matter functions} below. We have now specified all quantities appearing in equations~\eqref{E: HDprime}--\eqref{E: hc}. Hence, these represent the reduced, four-dimensional, constrained and closed dynamical system of Kantowski-Sachs cosmologies with Vlasov matter.

Finally we note, that the dynamical system~\eqref{E: HDprime}--\eqref{E: hc} is invariant under the discrete transformation
\begin{align}\label{E: disc sym}
(\tau,H_D,\Sigma_+) \to -(\tau,H_D,\Sigma_+) \,.
\end{align}

\subsection{The Vlasov matter functions}\label{SS: matter functions}

We work with Vlasov distributions $f_0$ which are invariant under the Kantowski-Sachs symmetries in the sense of Definition~\ref{D: 1}, and hence have the general functional form~\eqref{E: f0 LRS}. In addition we require $f_0$ to be reflection symmetric; cf Definition~\ref{D: 2}. From~\eqref{E: T Vlasov} we then have a diagonal energy-momentum tensor with
\begin{align}
\rho&\phantom{:}=T_{tt} = (\det g)^{-\tfrac{1}{2}} \int f_0\, \big(m+g^{11}v_1^2+g^{22}(v_2^2+v_3^2)\big)^{\tfrac{1}{2}}\,\mathrm dv_1\mathrm dv_2\mathrm dv_3 , \label{E: rho} \\
p_i&:={T^i}_i = (\det g)^{-\tfrac{1}{2}} \int f_0\,g^{ii}v_i^2 \big(m+g^{11}v_1^2+g^{22}(v_2^2+v_3^2)\big)^{-\tfrac{1}{2}}\,\mathrm dv_1\mathrm dv_2\mathrm dv_3 , \label{E: pi}
\end{align}
without summing over $i$; cf~\cite[Sec~2]{RendallUggla2000} or~\cite[p~1248]{CalogeroHeinzle2010}. Note that these expressions are functions of the inverse metric components, and that $p_2=p_3$. We now restrict to $m=1$ and elaborate below how massless Vlasov matter still fits into the presented framework. Following~\cite{CalogeroHeinzle2010} we define dimensionless principal pressures $w_i:=p_i/\rho$ and a dimensionless isotropic pressure $w:=\sum_i w_i/3$. Note that $w_2=w_3$. Instead of by the inverse metric components, we express these as functions of $l$ and $s$ which yields
\begin{align}\label{E: wi}
w_1(l,s) = (1-l)(1-2s)F_1(l,s) \quad\text{and}\quad
w_2(l,s) = (1-l)sF_2(l,s) \, ,
\end{align}
with $F_i(l,s)$ given by
\begin{align}\notag
\frac
{\int f_0\,v_i^2\Big(l\big(s^2(1-2s)\big)^{1/3}+(1-l)\big((1-2s)v_1^2+s(v_2^2+v_3^2)\big)\Big)^{-1/2}\mathrm dv_1\mathrm dv_2\mathrm dv_3}
{\int f_0\Big(l\big(s^2(1-2s)\big)^{1/3}+(1-l)\big((1-2s)v_1^2+s(v_2^2+v_3^2)\big)\Big)^{1/2}\mathrm dv_1\mathrm dv_2\mathrm dv_3} \, .
\end{align}
To assure that these are smooth functions, we have to assume that $f_0$ has split support (Definition~\ref{D: 3}), such that the denominator of $F_i$ is strictly positive. It is the functions~\eqref{E: wi} through which the Vlasov matter enters the dynamical system~\eqref{E: HDprime}--\eqref{E: hc}.

Though we formulated the $w_i$ for $m=1$, it is straightforward to show that we recover the expressions corresponding to $m=0$ if we formally set $l=0$ in~\eqref{E: wi}; cf~\cite[Sec~12.1]{CalogeroHeinzle2011}. $l=0$ thus marks the region which corresponds to solutions with massless Vlasov matter. In Section~\ref{S: state-space} we show that this is a boundary of the state-space. From the perspective of physics this is expected since in~\eqref{E: D}, $l$ is defined via the spatial volume element $\det g$. Thus towards a singularity, ie for $l\to0$ $\Leftrightarrow$ $\det g\to0$, one would expect massive particles to pick up momentum, such that the particle four-momenta become lightlike in the limit; cf~\cite[Sec~1]{RendallTod1999}. For later use we give the values
\begin{align}
w_1(0,1/2)&=0	&	w_2(0,1/2)&=1/2	&	w(0,s)&=1/3 \label{E: w0} 
\end{align}
which follow directly from~\eqref{E: wi}.

In analogy to the preceding argument, solutions approaching a state of infinite expansion, ie for $l\to1$ $\Leftrightarrow$ $\det g\to\infty$, one would expect particles to loose momentum, and to asymptotically resemble dust, ie a perfect fluid with equation of state $p=0$; cf~\cite{RendallUggla2000}. Indeed from~\eqref{E: wi} we find $w_1(1,s)$ $=$ $w_2(1,s)$ $=$ $w(1,s)$ $=$ $0$, ie an isotropic state with pressures 0. $l=1$ thus marks the region which corresponds to dust solutions. As we will see in Section~\ref{S: state-space} this is also a boundary of the state-space.

\section{The state-space}\label{S: state-space}

This section discusses the Kantowski-Sachs state-space and its features. In Subsection~\ref{SS: state-space X} we give the state-space corresponding to the dynamical system discussed in Section~\ref{S: DS for KS}. Subsection~\ref{SS: challenges} discusses two drawbacks: it reaches out to infinity in one coordinate and the dynamical system does not yield a smooth extension onto one of its boundary subsets. We point out how we deal with these issues, and define alternative coordinates in Subsection~\ref{SS: state-space Y}

\subsection{The state-space $\mathcal X$}\label{SS: state-space X}
To determine the state-space, we have to take into account the constraints which are imposed onto the dynamical quantities. Firstly, the range of $H_D$ follows directly from the definition of $D$, cf~\eqref{E: D}. Secondly, together with the requirement that the energy density should be positive, ie $\Omega>0$, the Hamiltonian constraint~\eqref{E: hc} determines the range of the shear parameter $\Sigma_+$. Finally, the ranges of $M_1$ and $l$ follow directly from their definitions~\eqref{E: M1 and l}, and in the case of $M_1$ also using the positivity of $D$, cf~\eqref{E: D}.

Let $x:=(H_D,\Sigma_+,M_1,l)$. The Kantowski-Sachs state-space is given by
\begin{align}\label{E: state-space}
\mathcal X:=\{x\in\R^4 | H_D\in(-1,1),\Sigma_+\in(-1,1),M_1>0,l\in(0,1)\}
\end{align}
and we denote its relevant boundary subsets as follows:
\begin{align}
\mathcal X_{i} &:= \{ x\in\partial\mathcal X | l=i \}, \label{E: X0} \quad i=0,1 \\
\mathcal X^\pm &:= \{ x\in\partial\mathcal X | H_D=\pm1 \} \notag \\
\mathcal X_{i}^\pm &:= \{ x\in\partial\mathcal X | H_D=\pm1, l=i \}, \notag \quad i=0,1\\
\mathcal I^\pm &:= \{ x\in\partial{\mathcal X} | H_D=\pm1, M_1=0 \} \notag \\
\mathcal I_0^\pm &:= \{ x\in\partial{\mathcal X} | H_D=\pm1, M_1=0, l=0 \notag \}
\end{align}
In this notation we follow the scheme that a subscript $0$ or $1$ refers to a restriction to $l=0$  or $1$, and a superscript $\pm$ to a restriction to $H_D=\pm1$. $\mathcal X$ corresponds to solutions with massive Vlasov matter. In Subsection~\ref{SS: matter functions} we established that $\mathcal X_0$ corresponds to massless Vlasov matter solutions and $\mathcal X_1$ to dust solutions. $\mathcal X_0$ is depicted in Figure~\ref{F: state-space X0}. From the Hamiltonian constraint~\eqref{E: hc} we see that the $\Sigma_+=\pm1$ boundaries correspond to vacuum solutions. Solutions in $\mathcal I^\pm$ are of LRS Bianchi type~I; cf~\cite{CalogeroHeinzle2010,CalogeroHeinzle2011}.
\begin{figure}
\tdplotsetmaincoords{70}{-38}
\begin{tikzpicture}[scale=3.3,tdplot_main_coords]
	
	\draw[dotted] (0,-1,0) -- (0,-1,3) -- (0,1,3);
	\draw[dotted] (0,1,3) -- (0,1,0) -- (0,-1,0);
	\fill[color=blue, opacity=.05] (0,-1,0) -- (0,-1,3) -- (0,1,3) -- (0,1,0) -- cycle;
	
	\draw[->] (-.8,-1.2,-.2) -- (.8,-1.2,-.2);
	\node at (-.8,-1.3,-.3) {$\scriptstyle -1$};
	\node at (.8,-1.3,-.3) {$\scriptstyle 1$};
	\node at (0,-1.3,-.3) {$\scriptstyle H_D$};
	\draw[->] (1.2,-1.2,.2) -- (1.2,-1.2,2.8);
	\node at (1.3,-1.3,.2) {$\scriptstyle 0$};
	\node at (1.3,-1.3,2.8) {$\scriptstyle\infty$};
	\node at (1.3,-1.3,1.5) {$\scriptstyle M_1$};
	\draw[->] (-1.2,-.8,-.2) -- (-1.2,.8,-.2);
	\node at (-1.3,-.8,-.3) {$\scriptstyle-1$};
	\node at (-1.3,.8,-.3) {$\scriptstyle1$};
	\node at (-1.3,0,-.3) {$\scriptstyle\Sigma_+$};
	\draw[dotted] (1,-1.05,0) -- (1,-2.1,0);
	\draw (1,-2,-.1) -- (1,-2,.3);
	\draw (1.1,-2,0) -- (.7,-2,0);
	\draw[->] (.8,-2,0) .. controls+(0,0,.12) and+(-.12,0,0) .. (1,-2,.2);
	\node at (.8,-2.1,0) {$\scriptstyle0$};
	\node at (.8,-1.9,0) {$\scriptstyle0$};
	\node at (1,-2.13,.2) {$\scriptstyle\pi/2$};
	\node at (1,-1.87,.2) {$\scriptstyle1/2$};
	\node at (.86,-2.1,.14) {$\scriptstyle\theta$};
	\node at (.86,-1.9,.14) {$\scriptstyle s$};
	
	\node at (1.1,-1,0) {$\scriptstyle T$};
	\node at (1,1,-.1) {$\scriptstyle Q$};
	\node at (1,.5,-.1) {$\scriptstyle R_\sharp$};
	\node at (1.1,-1.1,3.1) {$\scriptstyle T_\infty$};
	\node at (1.1,1.1,3.1) {$\scriptstyle Q_\infty$};
	\node at (1.1,.5,3.1) {$\scriptstyle R_\infty$};
	\node at (-1.1,1.1,-.1) {$\scriptstyle\tilde T$};
	\node at (-1.1,-1.1,-.1) {$\scriptstyle\tilde Q$};
	\node at (-1,-.5,-.1) {$\scriptstyle\tilde R_\sharp$};
	\node at (-1.1,1.1,3.1) {$\scriptstyle\tilde T_\infty$};
	\node at (-1,-1,3.1) {$\scriptstyle\tilde Q_\infty$};
	\node at (-1,-.5,3.1) {$\scriptstyle\tilde R_\infty$};
	
	\node at (-1.2,.5,3.5) {$\scriptstyle\mathcal X_0^-$};
	\draw (-1.2,.5,3.4) .. controls+(0,0,-.3) and+(-.3,0,0) .. (-1,.5,2.8);
	\node at (1.2,.2,3.3) {$\scriptstyle\mathcal X_0^+$};
	\draw (1.2,.2,3.2) .. controls+(0,0,-.2) and+(.2,0,0) .. (1.1,.2,2.8);
	\node at (-1.08,.15,-.08) {$\scriptstyle\mathcal I_0^+$};
	\draw (-1.08,.31,-.08) .. controls+(0,.2,0) and+(-.2,0,0) .. (-1.05,.5,0);
	\node at (.92,.-.45,-.08) {$\scriptstyle\mathcal I_0^+$};
	\draw (.92,-.29,-.08) .. controls+(0,.2,0) and+(-.2,0,0) .. (.92,-.1,0);
	\node at (-.2,.5,3.5) {$\scriptstyle H_D=0$};
	\draw (-.2,.5,3.4) .. controls+(0,0,-.2) and+(-.2,0,0) .. (0,.5,2.8);
	
	\draw (-1,1,0) -- (-1,1,1.5);\draw[<-] (-1,1,1.5) -- (-1,1,3);
	\draw (-1,1,3) -- (-1,.25,3);\draw[<-] (-1,.25,3) -- (-1,-.5,3);
	\draw (-1,-1,3) -- (-1,-.75,3);\draw[<-] (-1,-.75,3) -- (-1,-.5,3);
	\draw (-1,-1,3) -- (-1,-1,1.5);\draw[<-] (-1,-1,1.5) -- (-1,-1,0);
	\draw[thick, dotted] (-1,1,0) -- (-1,-1,0);
	\draw[dashed] (-1,-.5,3) -- (-1,-.5,1.5);
	\draw[<-,dashed] (-1,-.5,1.5) -- (-1,-.5,0);
	\draw[-<] (-1,-1,3) .. controls+(0,.3,-.1) and+(0,-.02,.4) .. (-1,-.75,1.5);
	\draw (-1,-.75,1.5) .. controls+(0,.02,.4) and+(0,-.3,.1) .. (-1,-.5,0);
	\draw[-<] (-1,1,0) .. controls+(0,-.5,.3) and+(0,.8,0) .. (-1,.25,2.5);
	\draw (-1,.25,2.5) .. controls+(0,-.8,0) and+(0,.5,.3) .. (-1,-.5,0);
	\draw[-<] (-1,1,0) .. controls+(0,-.8,.1) and+(0,.5,0) .. (-1,.25,1.5);
	\draw (-1,.25,1.5) .. controls+(0,-.5,0) and+(0,.8,.1) .. (-1,-.5,0);
	
	\draw[->,help lines] (1,1,0) -- (0,1,0);\draw[help lines] (0,1,0) -- (-1,1,0);
	\draw[->] (1,-1,0) -- (0,-1,0);\draw (0,-1,0) -- (-1,-1,0);
	\draw[->] (1,1,3) -- (0,1,3);\draw (0,1,3) -- (-1,1,3);
	\draw[->] (1,-1,3) -- (0,-1,3);\draw (0,-1,3) -- (-1,-1,3);
	
	\draw[->] (1,-1,0) -- (1,-1,1.5);\draw (1,-1,1.5) -- (1,-1,3);
	\draw[->] (1,-1,3) -- (1,-.25,3);\draw (1,-.25,3) -- (1,.5,3);
	\draw[->] (1,1,3) -- (1,.75,3);\draw (1,.75,3) -- (1,.5,3);
	\draw[->,help lines] (1,1,3) -- (1,1,1.5);\draw[help lines] (1,1,1.5) -- (1,1,0);
	\draw[help lines, thick, dotted] (1,1,0) -- (1,-1,0);
	\draw[->,dashed,help lines] (1,.5,3) -- (1,.5,1.5);
	\draw[dashed,help lines] (1,.5,1.5) -- (1,.5,0);
	\draw[->,help lines] (1,1,3) .. controls+(0,-.3,-.1) and+(0,.02,.4) .. (1,.75,1.5);
	\draw[help lines] (1,.75,1.5) .. controls+(0,-.02,-.4) and+(0,.3,.1) .. (1,.5,0);
	\draw[->,help lines] (1,-1,0) .. controls+(0,.5,.3) and+(0,-.8,0) .. (1,-.25,2.5);
	\draw[help lines] (1,-.25,2.5) .. controls+(0,.8,0) and+(0,-.5,.3) .. (1,.5,0);
	\draw[->,help lines] (1,-1,0) .. controls+(0,.8,.1) and+(0,-.5,0) .. (1,-.25,1.5);
	\draw[help lines] (1,-.25,1.5) .. controls+(0,.5,0) and+(0,-.8,.1) .. (1,.5,0);
	
\end{tikzpicture}
\caption{The three-dimensional $l=0$ boundary $\mathcal X_0$ of the four-dimensional state-space $\mathcal X$. Note that $M_1$ in unbounded. Solutions is $\mathcal X_0$ represent solutions with massless Vlasov matter.}
\label{F: state-space X0}
\end{figure}
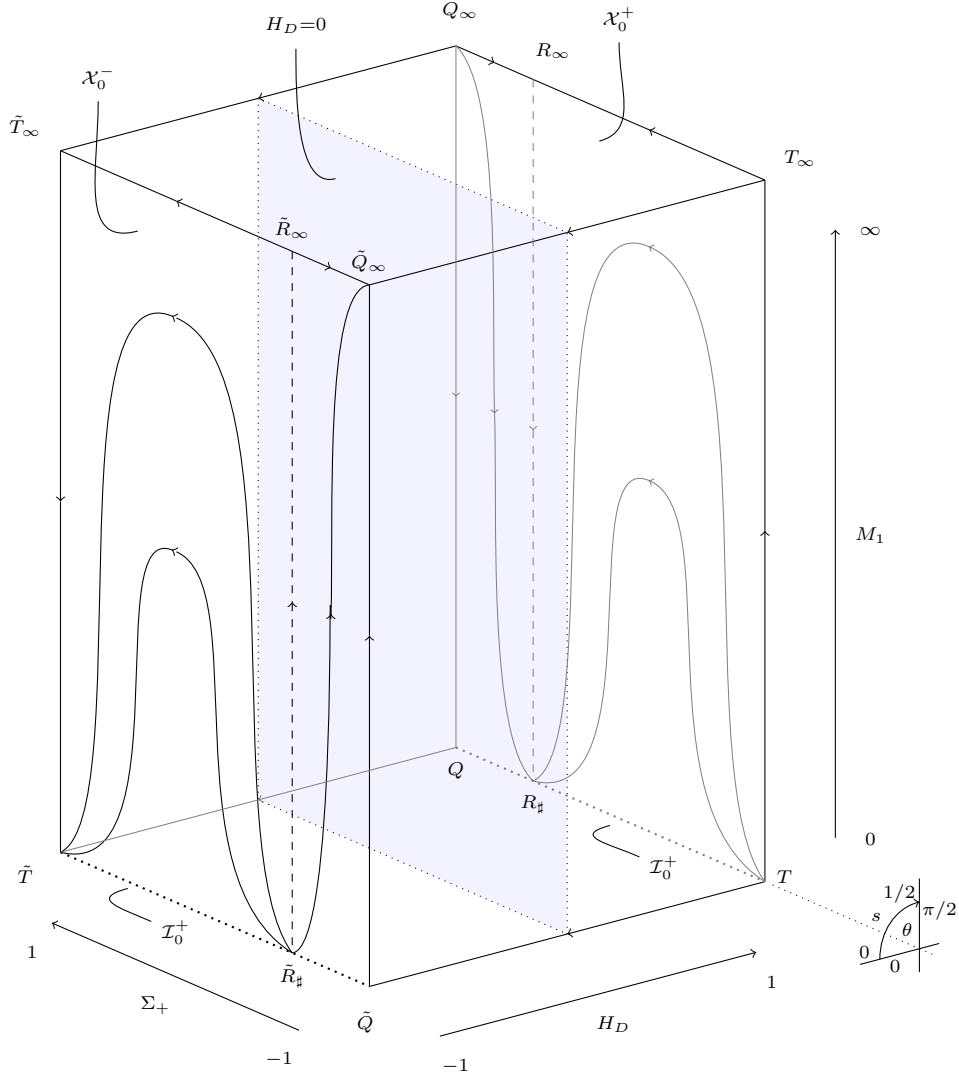 

\subsection{Main challenges}\label{SS: challenges}

There are two main challenges with the system~\eqref{E: HDprime}--\eqref{E: hc} on the state-space $\mathcal X$. Firstly, $\mathcal X$ is unbounded in $M_1$-direction. Thus we need to perform a careful analysis of the \emph{flow at infinity}; cf~\cite[Sec~3.10]{Perko2001}. In analogy to the approach taken in~\cite{Heissel2012}, we achieve this by a formal compactification via a projection of the flow onto the `Poincaré cylinder'; cf Section~\ref{S: results}, together with Appendix~\ref{A: flow at infinity}. Secondly, the dynamical system~\eqref{E: HDprime}--\eqref{E: hc} does not have a smooth extension onto $\mathcal I^\pm$. This is because the limit of $s$ for $H_D\to\pm1$ and $M_1\to0$ simultaneously does not exist; cf~\eqref{E: s}. Note however that each of the limits exists separately, so there is a smooth extension onto $\mathcal I^\pm$ on the boundary $\mathcal X^\pm$, and on the $M_1=0$ boundary. In the same way as in~\cite{CalogeroHeinzle2010, CalogeroHeinzle2011, Heissel2012} we deal with this issue by introducing polar coordinates centered at $H_D=\pm1$, $M_1=0$. In the new coordinates a smooth extension of the system onto the corresponding boundary then does exist; cf Subsection~\ref{SS: state-space Y}.

\subsection{The state-space $\mathcal Y$}\label{SS: state-space Y}

We introduce new coordinates $y:=(r,\theta,\Sigma_+,l)$ in order to analyse the solutions in a neighbourhood of $\mathcal I^\pm$. The coordinate transformation is given by
\begin{align}\label{E: coord trafo}
\frac{1-H_D^2}{2} = r\cos\theta	\qquad\text{and}\qquad		\frac{M_1^2}{3} &= r\sin\theta \,,
\end{align}
while $\Sigma_+$ and $l$ stay unchanged, and we choose $H_D>0$. These are polar coordinates centered at $\mathcal I^+$ which cover the $H_D>0$ region of $\mathcal X$; cf Figure~\ref{F: state-space X0}. The solutions in a neighbourhood of $\mathcal I^-$ then follow from the discrete symmetry~\eqref{E: disc sym}. With~\eqref{E: coord trafo}, \eqref{E: s} defines a bijection $s(\theta)$ which maps $(0,\pi)$ onto $(0,1/2)$. Hence we can identify the two.

We denote the state-space in the new coordinates by
\begin{align}\notag
\mathcal Y := \big\{ y\in\R^4 | r\in\big(0,\tfrac{1}{2\cos\theta}\big), \theta\in(0,\pi),\Sigma_+\in(-1,1),l\in(0,1) \big\} 
\end{align}
and denote its relevant boundary subsets as follows:
\begin{align*}
\mathcal Y_0 &:= \{ y\in\partial\mathcal Y | l=0 \} \\
\mathcal S &:= \{ y\in\partial\mathcal Y | \theta=\pi/2 \} &
\mathcal S_0 &:= \{ y\in\partial\mathcal Y | l=0, \theta=\pi/2 \} \\
\mathcal I^{\mathcal Y} &:= \{ y\in\partial\mathcal Y | r=0 \} &
\mathcal I_0^{\mathcal Y} &:= \{ y\in\partial\mathcal Y | l=0, r=0 \}
\end{align*}
From~\eqref{E: coord trafo} we have the following identificatons ($\sim$) between subsets of $\overline{\mathcal X}$ and $\overline{\mathcal Y}$:
\begin{align}\label{E: correspondences}
\mathcal Y&\sim\mathcal X\text{ for } H_D>0 , &
\overline{\mathcal S}\cup\overline{\mathcal I^\mathcal Y}&\sim\overline{\mathcal X^+}, &
\overline{\mathcal I^\mathcal Y}&\sim\overline{\mathcal I^+} ,
\end{align}
and equivalently for the respective subsets with subscript $0$. Note that the Bianchi~I boundary $\mathcal I^\mathcal Y$ is of one dimension higher than its counterpart in the $x$~coordinates. Because of this degeneracy this correspondence is not a diffeomorphism. However, \eqref{E: coord trafo} defines a diffeomorphism ($\cong$) between $\overline{\mathcal Y}\backslash\overline{\mathcal I^\mathcal Y}$ and $\overline{\mathcal X}\backslash\overline{\mathcal I^+}$ for $H_D>0$. Consequently the flows in these sets are topologically equivalent. In particular
\begin{align}\label{E: diffeo}
\mathcal S_0\cong\mathcal X_0^+ ,
\end{align}
so we can choose to work in either coordinate system to analyse the flow in these subsets.

While the $y$~coordinates allow us to analyse the flow in a neighbourhood of the Bianchi~I boundary, they have two drawbacks. Firstly, they only cover the $H_D>0$ region of the full state-space. Though by the discrete symmetry~\eqref{E: disc sym} this then also covers the $H_D<0$ region, it does not cover the $H_D=0$ region. Secondly, in the $y$~coordinates the region where $M_1\to\infty$ is of one dimension lower and thus degenerate. Hence, the $y$~coordinates are not suited as a basis to analyse the flow at infinity. In our analysis we thus only use the $y$ coordinates to analyse the flow in $\mathcal I_0^\mathcal Y$.
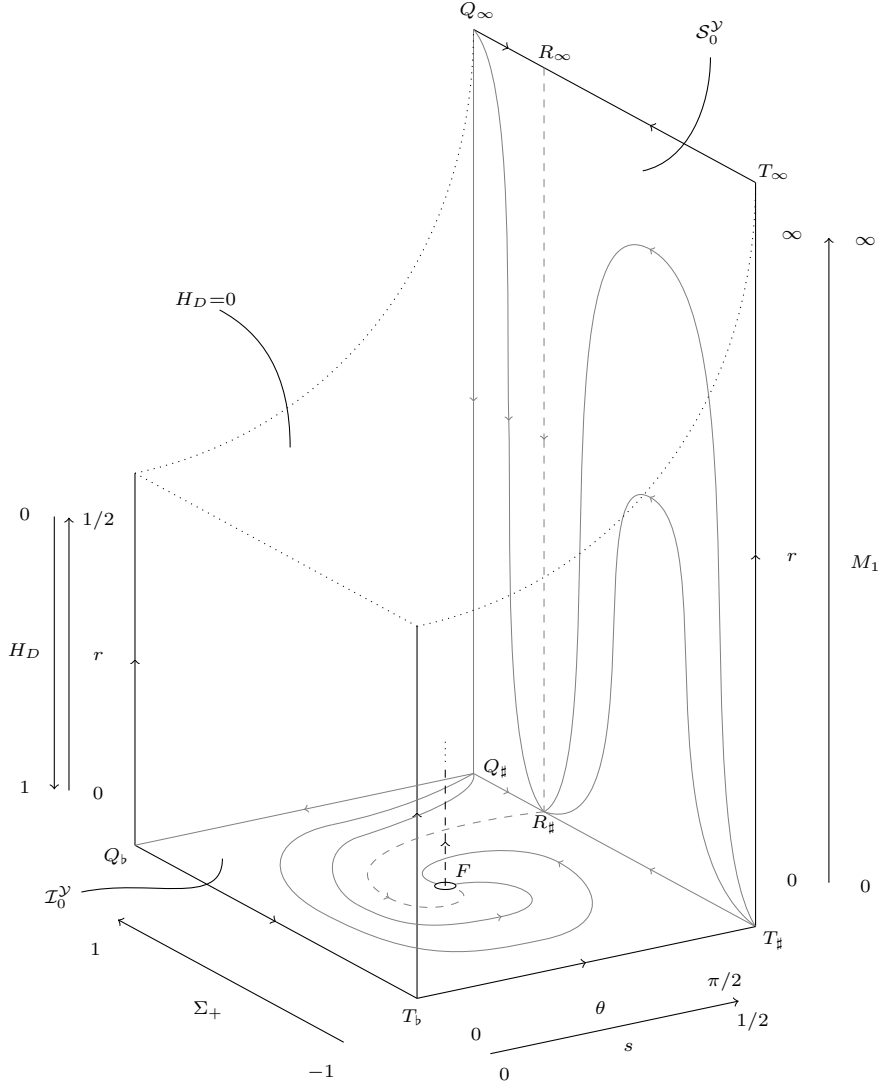
\begin{figure}
\tdplotsetmaincoords{70}{-32}
\begin{tikzpicture}[scale=3.5,tdplot_main_coords]
	
	\draw[->] (0,-1,0) -- (0,-1,1.5);\draw (0,-1,1.5) -- (0,-1,3);
	\draw[->] (0,-1,3) -- (0,-.25,3);\draw (0,-.25,3) -- (0,.5,3);
	\draw[->] (0,1,3) -- (0,.75,3);\draw (0,.75,3) -- (0,.5,3);
	\draw[->,help lines] (0,1,3) -- (0,1,1.5);\draw[help lines] (0,1,1.5) -- (0,1,0);
	\draw[->,help lines] (0,1,0) -- (0,.75,0);\draw[help lines] (0,.75,0) -- (0,.5,0);
	\draw[->,help lines] (0,-1,0) -- (0,-.25,0);\draw[help lines] (0,-.25,0) -- (0,.5,0);
	\draw[->,dashed,help lines] (0,.5,3) -- (0,.5,1.5);
	\draw[dashed,help lines] (0,.5,1.5) -- (0,.5,0);
	\draw[->,help lines] (0,1,3) .. controls+(0,-.3,-.1) and+(0,.02,.4) .. (0,.75,1.5);
	\draw[help lines] (0,.75,1.5) .. controls+(0,-.02,-.4) and+(0,.3,.1) .. (0,.5,0);
	\draw[->,help lines] (0,-1,0) .. controls+(0,.5,.3) and+(0,-.8,0) .. (0,-.25,2.5);
	\draw[help lines] (0,-.25,2.5) .. controls+(0,.8,0) and+(0,-.5,.3) .. (0,.5,0);
	\draw[->,help lines] (0,-1,0) .. controls+(0,.8,.1) and+(0,-.5,0) .. (0,-.25,1.5);
	\draw[help lines] (0,-.25,1.5) .. controls+(0,.5,0) and+(0,-.8,.1) .. (0,.5,0);
	
	\draw[->] (-1.5,-1,0) -- (-1.5,-1,.75);\draw (-1.5,-1,.75) -- (-1.5,-1,1.5);
	\draw[->] (-1.5,1,0) -- (-1.5,1,.75);\draw (-1.5,1,.75) -- (-1.5,1,1.5);
	\draw[->] (-1.5,1,0) -- (-1.5,0,0);\draw (-1.5,0,0) -- (-1.5,-1,0);
	\draw[dotted] (-1.5,-1,1.5) -- (-1.5,1,1.5);
	
	\draw[dotted] (-1.5,-1,1.5) .. controls+(.9,0,0) and+(0,0,-.9) .. (0,-1,3);
	\draw[dotted] (-1.5,1,1.5) .. controls+(.9,0,0) and+(0,0,-.9) .. (0,1,3);
	
	\draw[->,help lines] (0,1,0) -- (-.75,1,0); \draw[help lines] (-.75,1,0) -- (-1.5,1,0);
	\draw[->] (-1.5,-1,0) -- (-.75,-1,0); \draw (-.75,-1,0) -- (0,-1,0);
	\draw[->,dashed, help lines] (0,.5,0) .. controls+(-.4,.1,0) and+(0,.5,0) .. (-1,0,0);
	\draw[dashed, help lines] (-1,0,0) .. controls+(0,-.3,0) and+(.1,-.2,0) .. (-.75,0,0);
	\draw[help lines] (0,1,0) .. controls+(-.2,-.1,0) and+(.3,0,0) .. (-.75,.8,0);
	\draw[help lines] (-.75,.8,0) .. controls+(-.3,0,0) and+(0,.4,0) .. (-1.3,0,0);
	\draw[help lines] (-1.3,0,0) .. controls+(0,-.7,0) and+(-.3,0,0) .. (-.75,-.7,0);
	\draw[->,help lines] (-.75,-.7,0) .. controls+(.3,0,0) and+(0,-.3,0) .. (-.25,0,0);
	\draw[help lines] (-.25,0,0) .. controls+(0,.4,0) and+(-.15,.3,0) .. (-.75,0,0);
	\draw[help lines] (0,1,0) .. controls+(-.1,-.2,0) and+(.3,.06,0) .. (-.75,.63,0);
	\draw[help lines] (-.75,.63,0) .. controls+(-.2,-.06,0) and+(0,.2,0) .. (-1.14,0,0);
	\draw[->, help lines] (-1.14,0,0) .. controls+(0,-.3,0) and+(-.3,0,0) .. (-.75,-.4,0);
	\draw[help lines] (-.75,-.4,0) .. controls+(.3,0,0) and+(.3,0,0) .. (-.75,0,0);
	
	\draw[fill=white] (-.75,0,0) circle (.04);
	\node at (-.75,0,0) [above right] {$\scriptstyle F$};
	\draw[->, dashed] (-.75,0,0) -- (-.75,0,.18);
	\draw[dashed] (-.75,0,.18) -- (-.75,0,.5);
	\draw[dotted] (-.75,0,.5) -- (-.75,0,.6);
	\node at (-1.55,-1.05,-.05) {$\scriptstyle T_\flat$};
	\node at (.05,-1.05,-.05) {$\scriptstyle T_\sharp$};
	\node at (-1.55,1.05,-.05) {$\scriptstyle Q_\flat$};
	\node at (.1,1,0) {$\scriptstyle Q_\sharp$};
	\node at (0,.5,-.05) {$\scriptstyle R_\sharp$};
	\node at (.05,-1.05,3.05) {$\scriptstyle T_\infty$};
	\node at (.05,1.05,3.05) {$\scriptstyle Q_\infty$};
	\node at (.05,.5,3.05) {$\scriptstyle R_\infty$};
	
	\draw[->] (-1.3,-1.2,-.2) -- (-.2,-1.2,-.2);
	\node at (-1.3,-1.1,-.15) {$\scriptstyle0$};
	\node at (-1.3,-1.3,-.25) {$\scriptstyle0$};
	\node at (-.75,-1.1,-.15) {$\scriptstyle\theta$};
	\node at (-.75,-1.3,-.25) {$\scriptstyle s$};
	\node at (-.2,-1.1,-.15) {$\scriptstyle\pi/2$};
	\node at (-.2,-1.3,-.25) {$\scriptstyle1/2$};
	\draw[->] (.2,-1.2,.2) -- (.2,-1.2,2.8);
	\node at (.1,-1.1,.2) {$\scriptstyle0$};
	\node at (.3,-1.3,.2) {$\scriptstyle0$};
	\node at (.1,-1.1,1.5) {$\scriptstyle r$};
	\node at (.3,-1.3,1.5) {$\scriptstyle M_1$};
	\node at (.1,-1.1,2.8) {$\scriptstyle\infty$};
	\node at (.3,-1.3,2.8) {$\scriptstyle\infty$};
	\draw[->] (-1.7,-.8,-.2) -- (-1.7,.8,-.2);
	\node at (-1.8,-.8,-.3) {$\scriptstyle-1$};
	\node at (-1.8,0,-.3) {$\scriptstyle\Sigma_+$};
	\node at (-1.8,.8,-.3) {$\scriptstyle1$};
	\draw[->] (-1.68,1.18,.2) -- (-1.68,1.18,1.3);
	\draw[<-] (-1.72,1.22,.2) -- (-1.72,1.22,1.3);
	\node at (-1.8,1.3,.2) {$\scriptstyle1$};
	\node at (-1.6,1.1,.2) {$\scriptstyle0$};
	\node at (-1.8,1.3,.75) {$\scriptstyle H_D$};
	\node at (-1.6,1.1,.75) {$\scriptstyle r$};
	\node at (-1.8,1.3,1.3) {$\scriptstyle0$};
	\node at (-1.6,1.1,1.3) {$\scriptstyle1/2$};
	
	\node at (.3,-.2,3.3) {$\scriptstyle\mathcal S_0^{\mathcal Y}$};
	\draw (.3,-.2,3.2) .. controls+(0,0,-.2) and+(.2,0,0) .. (0,-.2,2.8);
	\node at (-1,1.3,2.01) {$\scriptstyle H_D=0$};
	\draw (-1,1.2,2) .. controls+(0,-.3,0) and+(0,0,.3) .. (-1,.7,1.6);
	\node at (-1.9,.9,-.1) {$\scriptstyle\mathcal I_0^\mathcal Y$};
	\draw (-1.8,.9,-.1) .. controls+(.3,0,0) and+(0,0,-.2) .. (-1.3,.7,0);
	
\end{tikzpicture}
\caption{The three-dimensional $l=0$ boundary $\mathcal Y_0$ of the four-dimensional state-space $\mathcal Y$. It corresponds to the $H_D>0$ region in Figure~\ref{F: state-space X0}. The $H_D=0$ surface is in this sense not a boundary of the Kantowski-Sachs state-space, but rather marks the end of the $y$ coordinate patch.}
\label{F: state-space Y0}
\end{figure} 

\section{Results}\label{S: results}

Subsection~\ref{SS: main results} we formulate our main theorem, which gives the past and future asymptotic solutions, and we formulate four lemmas by which we proof it. The proofs of the lemmas are given in Subsections~\ref{SS: proof L1}--\ref{SS: proof L4}. Subsections~\ref{SS: proof L1} and~\ref{SS: proof L2} also contain some physical interpretation concerning recollapse and the approach of massless Vlasov solutions by massive ones towards the singularities. In this section we use standard terminology of dynamical systems theory. We refer to ~\cite[Chap~4]{WainwrightEllis1997} and~\cite{Perko2001} for a background.

\subsection{Main theorem}\label{SS: main results}

\begin{theorem}\label{T: 1}
Consider Kantowski-Sachs cosmologies (Definition~\ref{D: KS}) with Vlasov matter of massive or massless particles, with distribution functions satisfying the following assumptions,
\begin{enumerate}[label=\emph{(}\roman*\emph{)}]
\item invariance under the Kantowski-Sachs symmetries (Definition~\ref{D: 1}),
\item reflection symmetry (Definition~\ref{D: 2}),
\item split support (Definition~\ref{D: 3}).
\end{enumerate}
The following statements hold:
\begin{enumerate}[label=\emph{(}\alph*\emph{)}]
\item Generic solutions are past and future asymptotic to the non-flat LRS Kasner vacuum solution given by the metric~\eqref{E: SH LRS metric} with
\begin{align}\notag
g_{11}(t)=at^{-2/3} \quad\text{and}\quad
g_{22}(t)=bt^{4/3},
\end{align}
where $a$, $b$ are positive constants, and where the time is shifted such that the big-bang (big-crunch) occurs at $t=0$.
\item There are non-generic solutions which are past (future) asymptotic to a non-isotropic Bianchi~I matter solution.
\item There are non-generic solutions which are past (future) asymptotic to the flat Friedman matter solution.
\end{enumerate}
The metric to these solutions is given by~\eqref{E: SH LRS metric} with components listed in Table~\ref{Tbl: 1}.
\end{theorem}
\begin{proof}
In the preceding sections we established that we can describe the dynamics in this scenario for massive particles by the dynamical system~\eqref{E: HDprime}--\eqref{E: hc} on the state-space $\mathcal X$, \eqref{E: state-space}, and for massless particles by the restriction of this system to the boundary $\mathcal X_0$, \eqref{E: X0}. We can thus proof our statement in the framework of dynamical systems theory, and do so through Lemmas~\ref{L: 1}--\ref{L: 4} below.
\end{proof}

\begin{lemma}\label{L: 1}
All past asymptotic solutions satisfy $H_D=1$ and all future asymptotic solutions satisfy $H_D=-1$.
\end{lemma}
\begin{proof}
Cf~Subsection~\ref{SS: proof L1}.
\end{proof}
\begin{lemma}\label{L: 2}
All past and future asymptotic solutions satisfy $l=0$.
\end{lemma}
\begin{proof}
Cf~Subsection~\ref{SS: proof L2}.
\end{proof}
\begin{lemma}\label{L: 3}
The only past (future) attractor in $\mathcal X\cup\mathcal X_0$ is the equilibrium point $Q_\infty$ ($\tilde Q_\infty$) at infinity. $R_\infty$ repels orbits from a three-dimensional (two-dimensional) set in $\mathcal X$ ($\mathcal X_0$). $\tilde R_\infty$ attracts accordingly. $F$ repels orbits from a two-dimensional (one-dimensional) set in $\mathcal X$ ($\mathcal X_0$). The image of $F$ under~\eqref{E: disc sym} attracts accordingly.
\end{lemma} 
\begin{proof}
Cf~Subsection~\ref{SS: proof L3}.
\end{proof}
\begin{lemma}\label{L: 4}
The fixed points correspond to the exact solutions listed in Table~\ref{Tbl: 1}.
\end{lemma}
\begin{proof}
Cf~Subsection~\ref{SS: proof L4}
\end{proof}

\subsection{Proof of Lemma~\ref{L: 1}}\label{SS: proof L1}

Since Vlasov matter satisfies a non-negative pressures condition, cf~\cite[Sec~1]{Rendall2004}, we have $w(l,s)\geq0$ in $\overline{\mathcal X}$. With this we find from~\eqref{E: HDprime} that $H_D'<0$ $\forall\,x\in\overline{\mathcal X}\mathcal \backslash\overline{\mathcal X^\pm}$. In other words, $H_D$ is strictly monotonically decreasing along orbits in $\overline{\mathcal X}\mathcal \backslash \overline{\mathcal X^\pm}$. We therefore know that
\begin{align}\label{E: limit sets in S}
\alpha(\Gamma) \subseteq \overline{\mathcal X^+} \quad\text{and}\quad
\omega(\Gamma) \subseteq \overline{\mathcal X^-} \qquad
\forall\,\Gamma\in\overline{\mathcal X}\mathcal \backslash \overline{\mathcal X^\pm} ,
\end{align}
where $\Gamma$ denotes an arbitrary orbit, and $\alpha(\Gamma)$ ($\omega(\Gamma)$) its $\alpha$ ($\omega$) limit set, ie its past (future) asymptotic sets; cf~\cite[Def~4.12]{WainwrightEllis1997}.

Note that~\eqref{E: limit sets in S} does not exclude the possibility that the limit sets could be empty. In other words, the solutions may satisfy $M_1\to\infty$ asymptotically. We thus have to analyse the flow at infinity. Since the only equation of the system~\eqref{E: HDprime}--\eqref{E: hc} which depends on the coordinate $M_1$ is the evolution equation of $M_1$ itself, naively we would await $H_D$ to also be strictly monotonically decreasing at infinity. A careful check via a formal compactification by projecting the flow onto the `Poincare cylinder' entails that this intuition is indeed correct; cf Appendix~\ref{A: flow at infinity}. Hence, we know that if the limit sets~\eqref{E: limit sets in S} are empty, ie if the solutions satisfy $M_1\to\infty$ for $\tau\to\pm\infty$, then also $H_D\to\pm1$. This completes the proof of Lemma~\ref{L: 1}.

\subsection*{Interpretation of Lemma~\ref{L: 1}}
We have shown that all Kantowski-Sachs cosmologies with Vlasov matter satisfy
\begin{align}\label{E: exp phase}
H_D\to+1\text{ for }\tau\to-\infty , \quad 
H_D\to0\text{ for }\tau=\tau^* , \quad
H_D\to-1\text{ for }\tau\to+\infty ,
\end{align}
where $\tau^*$ is some finite time, depending on the initial data. From the physical meaning of $H_D$ given in Section~\ref{SS: dynamical system} as the $D$-normalised Hubble scalar, we can interpret this result as follows: Kantowski-Sachs solutions with Vlasov matter undergo an expanding phase for $\tau<\tau*$, at the end of which they reach a state of maximal expansion for $\tau=\tau*$, after which they undergo a contracting phase for $\tau>\tau^*$.

\subsection{Proof of Lemma~\ref{L: 2}}\label{SS: proof L2}

For massless particles, ie for solutions in~$\mathcal X_0$, Lemma~\ref{L: 1} is trivial; cf Subsection~\ref{SS: matter functions}. From~\eqref{E: lprime} we see that for solutions in $\overline{\mathcal X}\backslash(\overline{\mathcal X_0}\cup\overline{\mathcal X_1})$ the sign of $l'$ is dictated by the sign of $H_D$. From~\eqref{E: exp phase} we thus know that Kantowski-Sachs solutions with massive Vlasov matter satisfy
\begin{align}\label{E: lprime at early times}
l'>0\text{ for }\tau\in(-\infty,\tau^*) , \quad
l'=0\text{ for }\tau=\tau^* , \quad
l'<0\text{ for }\tau\in(\tau^*,+\infty) , 
\end{align}
and together with~\eqref{E: limit sets in S} we can thus infer that
\begin{align}\label{E: limit sets in X0}
\alpha(\Gamma)\subseteq\overline{\mathcal X_0} \quad\text{and}\quad
\omega(\Gamma)\subseteq\overline{\mathcal X_0} \qquad
\forall\,\Gamma\in\overline{\mathcal X}\backslash(\overline{\mathcal X_0^\pm}\cup\overline{\mathcal X_1^\pm}) .
\end{align}
\eqref{E: limit sets in X0} does not exclude the possibility that the limit sets could be empty. Thus, as in the proof of Lemma~\ref{L: 1}, we have to analyse the flow at infinity. Following the same arguments, we would expect the monotonicities~\eqref{E: lprime at early times} to also hold at infinity, and in complete analogy, a projection of the flow onto the `Poincaré cylinder' shows that this is indeed the case; cf Appendix~\ref{A: flow at infinity}. Hence, we know that if the limit sets~\eqref{E: limit sets in X0} are empty, ie if the solutions satisfy $M_1\to\infty$ for $\tau\to\pm\infty$, then also $l\to0$. This completes the proof of Lemma~\ref{L: 2}.

\subsection*{Interpretation of Lemma~\ref{L: 2}}

From the physical interpretation of $l$ as a length scale associated with the spatial volume element $\det\,g$, cf Subsection~\ref{SS: dynamical system}, Lemma~\ref{L: 2} says the following: Kantowski-Sachs solutions with massive Vlasov matter exhibit a big-bang singularity towards the past and a big-crunch singularity towards the future; in other words, they recollapse. This is in compliance with more general recollapse results; cf Subsection~\ref{SS: recollapse}.

Furthermore, as discussed in Subsection~\ref{SS: matter functions}, $\mathcal X_0$, ie the $l=0$ boundary of $\mathcal X$, can be interpreted as the state-space for the case  of massless Vlasov particles, which is three-dimensional. In this sense we can conclude that Kantowski-Sachs cosmologies with massive Vlasov particles, behave like Kantowski-Sachs cosmologies with massless particles towards both, the big-bang and the big-crunch singularities.

\subsection{Proof of Lemma~\ref{L: 3}}\label{SS: proof L3}
So far we know~\eqref{E: limit sets in X0}, and that if these sets are empty, then the solutions must asymptote towards those regions of $M_1\to\infty$, for which $H_D=\pm1$ and $l=0$. What is left to do is to identify precisely where exactly solutions asymptote to in those two-dimensional regions. We restrict our analysis to the past asymptotics, from which we immediately get the future asymptotics as well via the discrete symmetry~\eqref{E: disc sym}.

From the second relation of~\eqref{E: correspondences} we see that in terms of the $y$~coordinates $\alpha(\Gamma)$ $\subseteq$ $\overline{\mathcal S_0}\cup\overline{\mathcal I_0^\mathcal Y}$. In $\mathcal I_0^\mathcal Y$ we have $M_1=0$, and for this the dynamical system~\eqref{E: HDprime}--\eqref{E: hc} coincides with that of LRS Bianchi~IX; cf~\cite[Sec~9]{CalogeroHeinzle2011}. Thus we can use the result of the analysis in~\cite[Sec~5]{CalogeroHeinzle2010}. The respective qualitative flow in $\overline{\mathcal I_0^\mathcal Y}$ is depicted in their Figure~3, and we reproduce it here in our Figure~\ref{F: flow diagram}.\footnote{We note that the flow on $\overline{\mathcal I_0^\mathcal Y}$ is also equivalent with~\cite[Fig~15(g)]{CalogeroHeinzle2011} and the respective parts of~\cite[Fig~7 and~8]{Heissel2012}; the latter in the context of LRS Bianchi type~VIII.}

Because of~\eqref{E: diffeo} the flow in $\mathcal S_0$ is qualitatively equivalent to the one in $\mathcal X_0^+$. We choose to analyse the latter, since the $x$~coordinates are better suited as a basis for the analysis of the flow at infinity than the $y$~coordinates; cf Subsection~\ref{SS: state-space Y}. To restrict the dynamical system~\eqref{E: HDprime}--\eqref{E: hc} to this boundary, we have to set $H_D=1$ and $l=0$. From~\eqref{E: s} we then have $s=1/2$, and from this and~\eqref{E: wi} we get~\eqref{E: w0}. With this the dynamical system on $\mathcal X_0^+$ reduces to
\begin{align}\label{E: DS in X0p}
\Sigma_+' &= -\frac{1}{2}(1-\Sigma_+^2)(2\Sigma_+-1)  , &
M_1' &= M_1 \big( \Sigma_+^2-4\Sigma_++1 \big) . 
\end{align}

This system exhibits three fixed points in $\overline{\mathcal X_0^+}$:
\begin{align}\label{E: fixed points on S}
T&:=(-1,0)	&	Q&:=(1,0)	&	R&:=(1/2,0)
\end{align}
Their local stability properties follow from the eigenvalues and eigenvectors of the linearisation of~\eqref{E: DS in X0p}, evaluated at~\eqref{E: fixed points on S}; cf~\cite[Sec~2.6]{Perko2001} or~\cite[Sec~4.3.2]{WainwrightEllis1997}. Below we follow the notation scheme $P:\begin{bmatrix}\begin{smallmatrix} \lambda_1 \\ \lambda_2 \end{smallmatrix}\end{bmatrix}, \begin{bmatrix} u_1 \, u_2 \end{bmatrix}$ where $\lambda_i$ denotes the eigenvalue corresponding to the eigenvector $u_i$ of the linearisation at $P$. We find
\begin{align*}
T&:\begin{bmatrix} 3 \\ 6 \end{bmatrix}, \begin{bmatrix} 1 & 0 \\ 0 & 1 \end{bmatrix}		&
&\rightarrow\text{local source,} \\
Q&:\begin{bmatrix} 1 \\ -2 \end{bmatrix}, \begin{bmatrix} 1 & 0 \\ 0 & 1 \end{bmatrix}		&
&\rightarrow\text{local saddle repelling in $\Sigma_+$-direction,} \\
R&:\begin{bmatrix} -3/4 \\ -3/4 \end{bmatrix}, \begin{bmatrix} 1 & 0 \\ 0 & 1 \end{bmatrix}	&
&\rightarrow\text{local sink.}
\end{align*}

Next we investigate the monotonicity of the flow. From~\eqref{E: DS in X0p}
\begin{align*}
\Sigma_+'&>0 \text{ for }\Sigma_+\in(-1,1/2) , \\
\Sigma_+'&=0 \text{ for }\Sigma_+\in\{\pm1,1/2\} , \\
\Sigma_+'&<0 \text{ for }\Sigma_+\in(1/2,1) .
\end{align*}
Therefore, the line $\Sigma_+=1/2$ represents a separatrix of the flow in $\mathcal X_0^+$; cf~\cite[Sec~3.11]{Perko2001}. For higher (lower) values of $\Sigma_+$, $\Sigma_+$ is strictly monotonically decreasing (increasing) along the flow in $\mathcal X_0^+$.

What is left to do is to investigate the flow at infinity. Since the first equation of~\eqref{E: DS in X0p} is independent of $M_1$, naively one would expect that the flow at infinity resembles that of the flow at $M_1=0$. With a projection of the flow onto the `Poincare cylinder' one can convince oneself that this is indeed the case; cf Appendix~\ref{A: flow at infinity}. This means that there are three fixed points at infinity,
\begin{align}\notag
T_\infty&:=(-1,\infty) \,,	&	Q_\infty&:=(1,\infty) \,,	&	R_\infty&:=(1/2,\infty) \,,
\end{align}
and the flow between these indeed resembles the flow on the line $M_1=0$.

We have now gathered all the information required to draw the qualitative flow diagram in $\mathcal X_0^+$ $\cong$ $\mathcal S_0$. It is shown in Figure~\ref{F: flow diagram} together with the flow in $\mathcal I_0^\mathcal Y$; cf also Figures~\ref{F: state-space X0} and~\ref{F: state-space Y0}. From this, together with the monotonicities of $H_D$ and $l$ obtained in Subsections~\ref{SS: proof L1} and~\ref{SS: proof L2}, we obtain Lemma~\ref{L: 3}.
\begin{figure}
\begin{tikzpicture}[scale=1.8]
	
	\draw[->] (0,-1) -- (0,-.25); \draw (0,-.25) -- (0,.75); \draw[<-] (0,.75) -- (0,1);	
	\draw[->] (0,1) -- (-.75,1); \draw (-.75,1) -- (-1.5,1);
	\draw[->] (-1.5,1) -- (-1.5,0); \draw (-1.5,0) -- (-1.5,-1);
	\draw[->] (-1.5,-1) -- (-.75,-1); \draw (-.75,-1) -- (0,-1);
	
	\draw[->] (0,-1) -- (2,-1); \draw (2,-1) -- (4,-1);	
	\draw[->] (4,-1) -- (4,-.25); \draw (4,-.25) -- (4,.75); \draw[<-] (4,.75) -- (4,1);
	\draw[->] (4,1) -- (2,1); \draw (2,1) -- (0,1);
	
	\draw[->,dashed] (4,.5) -- (2,.5); \draw[dashed] (2,.5) -- (0,.5);	
	\draw[->] (4,1) .. controls+(-135:.2) and+(0:.2) .. (2,.75);		
	\draw (2,.75) .. controls+(180:.2) and+(45:.2) .. (0,.5);
	\draw[->] (0,-1) .. controls+(45:.4) and+(-90:.7) .. (3.5,-.25);
	\draw (3.5,-.25) .. controls+(90:.7) and+(-45:.4) .. (0,.5);
	\draw[->] (0,-1) .. controls+(60:.5) and+(-90:.4) .. (1.5,-.25);
	\draw (1.5,-.25) .. controls+(90:.4) and+(-60:.5) .. (0,.5);

	\draw[->,dashed] (0,.5) .. controls+(170:.4) and+(95:.6) .. (-1,0);	
	\draw[dashed] (-1,0) .. controls+(-85:.6) and+(-25:.4) .. (-.75,0);
	\draw[->] (0,1) .. controls+(-170:.7) and+(95:.9) .. (-1.35,0); 
	\draw (-1.35,0) .. controls+(-85:.5) and+(180:.5) .. (-.75,-.8);
	\draw (-.75,-.8) .. controls+(0:.4) and+(-90:.6) .. (-.2,0);
	\draw (-.2,0) .. controls+(90:.6) and+(155:.3) .. (-.75,0);
	\draw[->] (0,1) .. controls+(-120:.3) and+(95:.8) .. (-1.18,0);
	\draw (-1.18,0) .. controls+(-85:.4) and+(-160:.3) .. (-.65,-.50);
	\draw (-.65,-.50) .. controls+(20:.3) and+(30:.5) .. (-.75,0);
	
	\draw[fill=white] (-.75,0) circle (.03);	
	\node[below] at (-.75,0) {$\scriptstyle F$};
	\draw[fill=white] (-1.5,1) circle (.03);
	\node[above left] at (-1.5,1) {$\scriptstyle Q_\flat$};
	\draw[fill=white] (-1.5,-1) circle (.03);
	\node[below left] at (-1.5,-1) {$\scriptstyle T_\flat$};
	\node[above] at (0,1) {$\scriptstyle Q_\sharp$};
	\node[below] at (0,-1) {$\scriptstyle T_\sharp$};
	\node[below left] at (0,.5) {$\scriptstyle R_\sharp$};
	\draw[fill=white] (4,1) circle (.03);
	\node[above right] at (4,1) {$\scriptstyle Q_\infty$};
	\draw[fill=white] (4,-1) circle (.03);
	\node[below right] at (4,-1) {$\scriptstyle T_\infty$};
	\draw[fill=white] (4,.5) circle (.03);
	\node[right] at (4,.5) {$\scriptstyle R_\infty$};
	
	\draw[->] (.2,-1.5) -- (3.8,-1.5);
	\node[below] at (.2,-1.5) {$\scriptstyle0$};
	\node[above] at (.2,-1.5) {$\scriptstyle0$};
	\node[below] at (3.8,-1.5) {$\scriptstyle\infty$};
	\node[above] at (3.8,-1.5) {$\scriptstyle\infty$};
	\node[below] at (2,-1.5) {$\scriptstyle M_1$};
	\node[above] at (2,-1.5) {$\scriptstyle r$};
	
	\draw[->] (-1.3,-1.5) -- (-.2,-1.5);
	\node[below] at (-1.3,-1.5) {$\scriptstyle0$};
	\node[above] at (-1.3,-1.5) {$\scriptstyle0$};
	\node[below] at (-.2,-1.5) {$\scriptstyle1/2$};
	\node[above] at (-.2,-1.5) {$\scriptstyle\pi$};
	\node[below] at (-.75,-1.5) {$\scriptstyle s$};
	\node[above] at (-.75,-1.5) {$\scriptstyle\theta$};

	\draw[->] (-2,-.8) -- (-2,.8);
	\node[left] at (-2,0) {$\scriptstyle\Sigma_+$};
	\node[left] at (-2,-.8) {$\scriptstyle-1$};
	\node[left] at (-2,.8) {$\scriptstyle1$};
	
	\node at (-.75,1.5) {$\scriptstyle\mathcal I_0^\mathcal Y$};
	\node at (2,1.5) {$\scriptstyle\mathcal S_0\cong\mathcal X_0^+$};

\end{tikzpicture}
\caption{The qualitative flow on $\overline{\mathcal I_0^\mathcal Y}\cup\overline{\mathcal S_0}$ $\sim$ $\overline{\mathcal X_0^+}$. As a consequence of Lemmas~\ref{L: 1} and~\ref{L: 2}, all fixed points are repelling orbits in the two respective orthogonal directions, which cannot be seen from the graph: $F$, $T_\flat$ and $Q_\flat$ are repelling in $r$ and $l$-direction, where for $T_\flat$ and $Q_\flat$, the $r$-direction also coincides with the $H_D$-direction. $T_\infty$ and $Q_\infty$ are repelling in $H_D$ and $l$-direction. $Q_\infty$ is the past attractor.}
\label{F: flow diagram}
\end{figure}
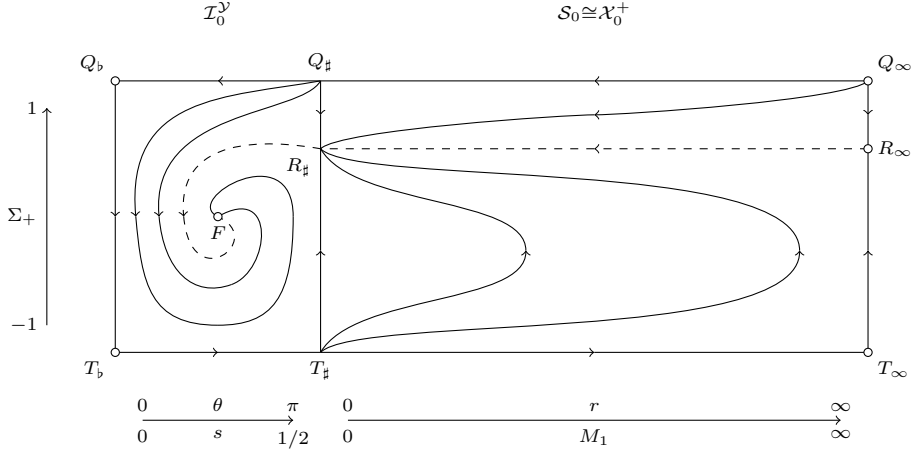 

\subsection*{Remark} Note that concerning massless Vlasov matter we recover~\cite[Thm~5.1]{RendallTod1999}.

\subsection{Proof of Lemma~\ref{L: 4}}\label{SS: proof L4}

We can identify the fixed points with exact solutions. \cite[App~A]{CalogeroHeinzle2011} gives a calculation of the components of the metric~\eqref{E: SH LRS metric} and the energy-momentum tensor~\eqref{E: rho}--\eqref{E: pi} in terms of the fixed point coordinates. The result for the metric components is
\begin{align}\notag
g_{11}(t) = at^{2\gamma_1} \qquad\text{and}\qquad g_{22}(t) = bt^{2\gamma_2}
\end{align}
with positive and generally constrained constaints $a,b$,
\begin{align}\notag
\gamma_1 = \frac{H_D-2\Sigma_+}{H_D(1+q)+\Sigma_+(1-H_D^2)} \qquad\text{and}\qquad
\gamma_2 = \frac{H_D+\Sigma_+}{H_D(1+q)+\Sigma_+(1-H_D^2)} .
\end{align}
The result for the components~\eqref{E: rho}--\eqref{E: pi} of the energy-momentum tensor is
\begin{align}\notag
\rho(t) = \frac{3(1-\Sigma_+^2)}{\big(H_D(1+q)+\Sigma_+(1-H_D^2)\big)^2}t^{-2} \qquad\text{and}\qquad
p_i(t) = w_i(l,s)\rho(t) .
\end{align}
Note that these expressions are only valid at the fixed points. Entering the coordinates for our fixed points and fixed points at infinity we obtain Table~\ref{Tbl: 1}, which concludes the proofs of Lemma~\ref{L: 4} and Theorem~\ref{T: 1}.
\begin{table}
\begin{tabular}{ l | c | c | c | c | c | l}
  fixed point & $g_{11}(t)$ & $g_{22}(t)$ & $\rho(t)$ & $p_1(t)$ & $p_2(t)$	& solution	\\ \hline
  $Q, Q_\flat, Q_\sharp,Q_\infty$	& $at^{-2/3}$	& $bt^{4/3}$	& $0$	& $0$	& $0$	& non-flat LRS Kasner\\
  $T, T_\flat, T_\sharp,T_\infty$		& $at^2$		& $b$		& $0$	& $0$	& $0$ 	& Taub (flat LRS Kasner)	\\
  $R, R_\sharp,R_\infty$ 			& $a$		& $bt^{4/3}$	& $\tfrac{4}{9}t^{-2}$	& $0$	& $\tfrac{2}{9}t^{-2}$	& no-name Bianchi~I		\\
  $F$							& $at$		& $bt$		& $\tfrac{3}{4}t^{-2}$	& $\tfrac{1}{4}t^{-2}$	& $\tfrac{1}{4}t^{-2}$	& flat Friedman
\end{tabular} \vspace{.4 cm}
\caption{The components of the metric~\eqref{E: SH LRS metric} and the energy-momentum tensor~\eqref{E: rho}--\eqref{E: pi} of the fixed point solutions. $a$ and $b$ are positive constants. The fixed points listed here satisfy $l=0$ and are depicted in Figures~\ref{F: state-space X0} and~\ref{F: state-space Y0}. However these have a copy at the $l=1$ boundary. Including the latter this table is complete.}
\label{Tbl: 1}
\end{table}


\begin{appendix}

\section{The flow at infinity}\label{A: flow at infinity}

In dynamical systems theory, the flow at infinity is usually analysed using projections of the flow onto the Poincaré sphere; cf~\cite[Sec~3.10]{Perko2001}. The sphere is the natural surface to project onto if one deals with a state-space which extends to infinity in all directions. However both, the state-space of LRS Bianchi~VIII in~\cite{Heissel2012} as well as that of Kantowski-Sachs in the present paper, extend to infinity in one coordinate direction only; the former in $H_D$-direction and the latter in $M_1$-direction. Hence, in these cases we seek to analyse the flow at infinity in the respective directions. However, those regions project to points on the Poincaré sphere, which is why this approach is not well suited for an analysis of the flow at infinity in these cases. In~\cite[Sec~6]{Heissel2012} the second author adopted this technique with the slight modification that the compactification is only done in one direction; ie projected is onto the `Poincaré cylinder' rather than onto the Poincaré sphere. In the appendix of that paper this has been presented in two and three dimensions. Though the generalisation to higher dimensions is trivial, for completeness, and to make the paper more self-contained, in the following we will present this technique in four dimensions as well.

Consider the dynamical system
\begin{align}\label{E: ds PQRS}
\begin{bmatrix}
H_D \\ \Sigma_+ \\ M_1 \\ l
\end{bmatrix}' =
\begin{bmatrix}
P(H_D,\Sigma_+,M_1,l) \\
Q(H_D,\Sigma_+,M_1,l) \\
R(H_D,\Sigma_+,M_1,l) \\
S(H_D,\Sigma_+,M_1,l) \\
\end{bmatrix} \, ,
\end{align}
where $P$, $Q$, $R$ and $S$ are polynomials of degree $p$, $q$, $r$ and $s$ in $M_1$, satisfying
\begin{align}\label{E: exponent condition}
p,q,s+1\geq r \, .
\end{align}
The system~\eqref{E: ds PQRS} can also be written as
\begin{align}
P(H_D,\Sigma_+,M_1,l)\,\mathrm dM_1 - R(H_D,\Sigma_+,M_1,l)\,\mathrm dH_D &= 0 \, , \label{E: 2a} \\
Q(H_D,\Sigma_+,M_1,l)\,\mathrm dM_1 - R(H_D,\Sigma_+,M_1,l)\,\mathrm d\Sigma_+ &= 0 \, , \label{E: 2b} \\
S(H_D,\Sigma_+,M_1,l)\,\mathrm dM_1 - R(H_D,\Sigma_+,M_1,l)\,\mathrm dl &= 0 \, , \label{E: 2c}
\end{align}
where however the information of the flow direction is lost. Supported by the illustration in Figure~\ref{F: Poincare cylinder}, we now project onto the Poincaré cylinder
\begin{align}\notag
\{ (H_D,\Sigma_+,X,l,Z)\in\R^5 \, | \, X^2+Z^2=1 \}
\end{align}
via the transformation
\begin{align*}
(H_D,\Sigma_+,M_1,l) &\to (H_D,\Sigma_+,X,l,Z) : X^2+Z^2=1 \, , \\
M_1 &= X/Z \, , \\
\mathrm dM_1 &= \frac{1}{Z}\mathrm dX - \frac{X}{Z^2}\mathrm dZ \, .
\end{align*}
Applying this to~\eqref{E: 2a} and multiplying by $Z^{p+2}$, we get
\begin{align}\label{E: 2ap}
Z^{p+1}P\,\mathrm dX - Z^pXP\,\mathrm dZ - Z^{p+2}R\,\mathrm dH_D = 0 \, .
\end{align}
The flow at infinity of the original state-space projects onto the flow on the `equator' $X=1$, $Z=0$ of the Poincaré cylinder. We thus evaluate~\eqref{E: 2ap} on the `equator', and get
\begin{equation}\notag
\big( Z^pP(H_D,\Sigma_+,X/Z,l) \big)_{X=1,Z=0}\,\mathrm dZ = 0 \, ,
\end{equation}
and performing the analogous steps for the other two equations~\eqref{E: 2b} and~\eqref{E: 2c}, we have
\begin{align*}
\big( Z^qQ(H_D,\Sigma_+,X/Z,l) \big)_{X=1,Z=0}\,\mathrm dZ &= 0 \, , \\
\big( Z^sS(H_D,\Sigma_+,X/Z,l) \big)_{X=1,Z=0}\,\mathrm dZ &= 0 \, .
\end{align*}
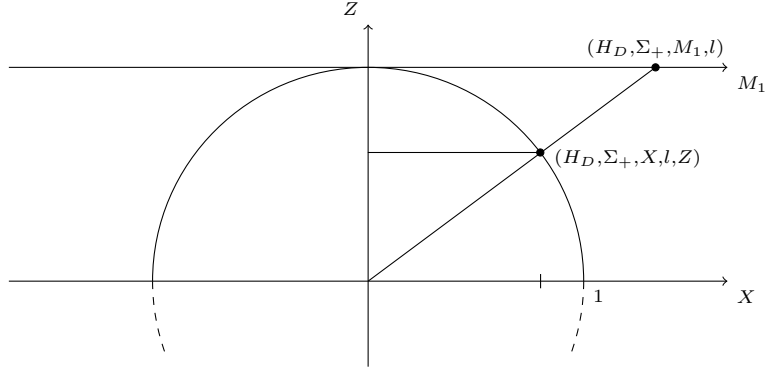
\begin{figure}\begin{tikzpicture}[scale=1.89]
	\draw[->](-2.5,0)--(2.5,0);		\node[below right]at(2.5,0){$\scriptstyle X$};		
	\draw[->](0,-.6)--(0,1.8);		\node[above left]at(0,1.8){$\scriptstyle Z$};		
	\draw[->](-2.5,1.5)--(2.5,1.5);	\node[below right]at(2.5,1.5){$\scriptstyle{M_1}$}; 	
	\draw(1.5,0)arc(0:180:1.5cm);				%
	\draw[dashed](1.5,0)arc(0:-20:1.5cm);		
	\draw[dashed](-1.5,0) arc (180:200:1.5cm);	%
	\node[below right]at(1.5,0){$\scriptstyle 1$};
	\draw(0,0)--(2,1.5);						
	\draw[fill] (2,1.5) circle (0.025);	\node[above]at(2,1.5){$\scriptstyle{(H_D,\Sigma_+,M_1,l)}$};	
	\draw[fill] (37:1.5) circle (0.025);\node[right]at(35:1.5){$\scriptstyle{(H_D,\Sigma_+,X,l,Z)}$};	
	\draw(37:1.5cm)--++(-1.2,0);\draw(1.2,-.05)--(1.2,.05);
\end{tikzpicture}
\caption{Projection onto the `Poincar\'e cylinder'. $M_1=X/Z$ by similar triangles.}
\label{F: Poincare cylinder}\end{figure}
The flow at the `equator' is then characterised by $\mathrm dZ=0$, and given by the three-dimensional dynamical system
\begin{align}\label{E: equator system}
\begin{bmatrix}
H_D \\ \Sigma_+ \\ l
\end{bmatrix}' =
\begin{bmatrix}
Z^pP(H_D,\Sigma_+,X/Z,l) \\
Z^qQ(H_D,\Sigma_+,X/Z,l) \\
Z^sS(H_D,\Sigma_+,X/Z,l)
\end{bmatrix}_{X=1,Z=0} \, .
\end{align}

In the case of the Kantowski-Sachs evolution equations~\eqref{E: HDprime} to~\eqref{E: lprime} we have
\begin{align}\label{E: pqs=0}
p,q,s=0
\end{align}
and $r=1$, hence~\eqref{E: exponent condition} is satisfied. Further, because of~\eqref{E: pqs=0} it follows that the system~\eqref{E: equator system} in this case simply corresponds to the original system with the evolution equation for $M_1$ being dropped, ie
\begin{align}\notag
\begin{bmatrix}
H_D \\ \Sigma_+ \\ l
\end{bmatrix}' =
\begin{bmatrix}
P(H_D,\Sigma_+,l) \\
Q(H_D,\Sigma_+,l) \\
S(H_D,\Sigma_+,l)
\end{bmatrix} \, ,
\end{align}
with $P$, $Q$ and $S$ given by the right hand sides of~\eqref{E: HDprime}, \eqref{E: Sigmaprime} and~\eqref{E: lprime}.

\end{appendix}

\bibliography{BibFile}{}

\begin{thebibliography}{10}

\bibitem{AnderssonFajman2017}
{\sc L.~Andersson and D.~Fajman}, {\em {N}onlinear stability of the {M}ilne
  model with matter}, arXiv preprint arXiv:1709.00267,  (2017).

\bibitem{Andreasson2011}
{\sc H.~Andr{\'e}asson}, {\em The {E}instein-{V}lasov system/kinetic theory},
  Living Reviews in Relativity, 14 (2011).

\bibitem{BerezdivinSachs1973}
{\sc R.~Berezdivin and R.~K. Sachs}, {\em {M}atter symmetries in general
  relativistic kinetic theory}, Journal of Mathematical Physics, 14 (1973),
  pp.~1254--1257.

\bibitem{Burnett1991}
{\sc G.~A. Burnett}, {\em {I}ncompleteness theorems for the spherically
  symmetric spacetimes}, Phys. Rev. D, 43 (1991), pp.~1143--1149.

\bibitem{CalogeroHeinzle2009}
{\sc S.~Calogero and J.~M. Heinzle}, {\em {D}ynamics of {B}ianchi {T}ype {I}
  {S}olutions of the {E}instein {E}quations with {A}nisotropic {M}atter},
  Annales Henri Poincar{\'e}, 10 (2009), pp.~225--274.

\bibitem{CalogeroHeinzle2010}
\leavevmode\vrule height 2pt depth -1.6pt width 23pt, {\em {O}scillations
  toward the singularity of locally rotationally symmetric {B}ianchi type {IX}
  cosmological models with vlasov matter}, SIAM Journal on Applied Dynamical
  Systems, 9 (2010), pp.~1244--1262.

\bibitem{CalogeroHeinzle2011}
\leavevmode\vrule height 2pt depth -1.6pt width 23pt, {\em {B}ianchi
  cosmologies with anisotropic matter: {L}ocally rotationally symmetric
  models}, Physica D: Nonlinear Phenomena, 240 (2011), pp.~636 -- 669.

\bibitem{Coley2003}
{\sc A.~A. Coley}, {\em {D}ynamical systems and cosmology}, vol.~291 of
  Astrophysics and Space Science Library, Kluwer Academic Publishers, 2003.

\bibitem{Collins1977}
{\sc C.~B. Collins}, {\em {G}lobal structure of the ''{K}antowski--{S}achs''
  cosmological models}, Journal of Mathematical Physics, 18 (1977),
  pp.~2116--2124.

\bibitem{EllisEtAl1983i}
{\sc G.~F.~R. Ellis, D.~R. Matravers, and R.~Treciokas}, {\em {A}nisotropic
  solutions of the {E}instein-{B}oltzmann equations: {I}. {G}eneral formalism},
  Annals of Physics, 150 (1983), pp.~455 -- 486.

\bibitem{FajmanEtAl2017}
{\sc D.~Fajman, J.~Joudioux, and J.~Smulevici}, {\em {T}he {S}tability of the
  {M}inkowski space for the {E}instein-{V}lasov system}, arXiv:1707.06141,
  (2017).

\bibitem{HeinzleUggla2006}
{\sc J.~M. Heinzle and C.~Uggla}, {\em {D}ynamics of the spatially homogeneous
  {B}ianchi type {I} {E}instein--{V}lasov equations}, Classical and Quantum
  Gravity, 23 (2006), p.~3463.

\bibitem{Heissel2012}
{\sc G.~Hei{\ss}el}, {\em {D}ynamics of locally rotationally symmetric
  {B}ianchi type {VIII} cosmologies with anisotropic matter}, General
  Relativity and Gravitation, 44 (2012), pp.~2901--2922.

\bibitem{Henkel2002}
{\sc O.~Henkel}, {\em {G}lobal prescribed mean curvature foliations in
  cosmological space--times. {I}}, Journal of Mathematical Physics, 43 (2002),
  pp.~2439--2465.

\bibitem{HorwoodEtAl2003}
{\sc J.~T. Horwood, M.~J. Hancock, D.~The, and J.~Wainwright}, {\em {L}ate-time
  asymptotic dynamics of {B}ianchi {VIII} cosmologies}, Classical and Quantum
  Gravity, 20 (2003), pp.~1757--1777.

\bibitem{LeeNungesser2017}
{\sc H.~Lee and E.~Nungesser}, {\em {F}uture global existence and asymptotic
  behaviour of solutions to the {E}instein--{B}oltzmann system with {B}ianchi
  {I} symmetry}, Journal of Differential Equations, 262 (2017), pp.~5425 --
  5467.

\bibitem{LeeNungesser2018}
\leavevmode\vrule height 2pt depth -1.6pt width 23pt, {\em {S}elf-{S}imilarity
  {B}reaking of {C}osmological {S}olutions with {C}ollisionless {M}atter},
  Annales Henri Poincar{\'e}, 19 (2018), pp.~2137--2155.

\bibitem{LindbladTaylor2017}
{\sc H.~Lindblad and M.~Taylor}, {\em {G}lobal stability of {M}inkowski space
  for the {E}instein--{V}lasov system in the harmonic gauge}, arXiv preprint
  arXiv:1707.06079,  (2017).

\bibitem{MaartensMaharaj1985}
{\sc R.~Maartens and S.~D. Maharaj}, {\em {C}ollision‐free gases in spatially
  homogeneous space‐times}, Journal of Mathematical Physics, 26 (1985),
  pp.~2869--2880.

\bibitem{MaartensMaharaj1990}
\leavevmode\vrule height 2pt depth -1.6pt width 23pt, {\em {C}ollision-free
  gases in {B}ianchi space-times}, General Relativity and Gravitation, 22
  (1990), pp.~595--607.

\bibitem{Nungesser2010}
{\sc E.~Nungesser}, {\em {I}sotropization of non-diagonal {B}ianchi {I}
  spacetimes with collisionless matter at late times assuming small data},
  Classical and Quantum Gravity, 27 (2010), p.~235025.

\bibitem{Nungesser2012}
\leavevmode\vrule height 2pt depth -1.6pt width 23pt, {\em {F}uture non-linear
  stability for solutions of the {E}instein-{V}lasov system of {B}ianchi types
  {II} and {VI}$_0$}, Journal of Mathematical Physics, 53 (2012), p.~102503.

\bibitem{Nungesser2013}
\leavevmode\vrule height 2pt depth -1.6pt width 23pt, {\em {F}uture
  {N}on-{L}inear {S}tability for {R}eflection {S}ymmetric {S}olutions of the
  {E}instein--{V}lasov {S}ystem of {B}ianchi {T}ypes {II} and {VI}$_0$},
  Annales Henri Poincar{\'e}, 14 (2013), pp.~967--999.

\bibitem{NungesserEtAl2013}
{\sc E.~Nungesser, L.~Andersson, S.~Bose, and A.~A. Coley}, {\em
  {I}sotropization of solutions of the {E}instein--{V}lasov system with
  {B}ianchi {V} symmetry}, General Relativity and Gravitation, 46 (2013),
  p.~1628.

\bibitem{Perko2001}
{\sc L.~Perko}, {\em {D}ifferential equations and dynamical systems}, vol.~7 of
  Texts in Applied Mathemtaics, Springer-Verlag New York, Inc., 3~ed., 2001.

\bibitem{Rendall1996}
{\sc A.~D. Rendall}, {\em {T}he initial singularity in solutions of the
  {E}instein--{V}lasov system of {B}ianchi type {I}}, Journal of Mathematical
  Physics, 37 (1996), pp.~438--451.

\bibitem{Rendall2002}
\leavevmode\vrule height 2pt depth -1.6pt width 23pt, {\em {C}osmological
  {M}odels and {C}entre {M}anifold {T}heory}, General Relativity and
  Gravitation, 34 (2002), pp.~1277--1294.

\bibitem{Rendall2004}
\leavevmode\vrule height 2pt depth -1.6pt width 23pt, {\em {T}he
  {E}instein-{V}lasov {S}ystem}, in The Einstein Equations and the Large Scale
  Behavior of Gravitational Fields, P.~T. Chru{\'{s}}ciel and H.~Friedrich,
  eds., Basel, 2004, Birkh{\"a}user Basel, pp.~231--250.

\bibitem{Rendall2008}
\leavevmode\vrule height 2pt depth -1.6pt width 23pt, {\em {P}artial
  differential equations in general relativity}, vol.~16 of Oxford graduate
  texts in mathematics, Oxford University Press, 2008.

\bibitem{RendallTod1999}
{\sc A.~D. Rendall and K.~P. Tod}, {\em {D}ynamics of spatially homogeneous
  solutions of the {E}instein-{V}lasov equations which are locally rotationally
  symmetric}, Classical and Quantum Gravity, 16 (1999), p.~1705.

\bibitem{RendallUggla2000}
{\sc A.~D. Rendall and C.~Uggla}, {\em {D}ynamics of spatially homogeneous
  locally rotationally symmetric solutions of the {E}instein-{V}lasov
  equations}, Classical and Quantum Gravity, 17 (2000), p.~4697.

\bibitem{Ringstroem2013}
{\sc H.~Ringstr{\"o}m}, {\em {O}n the topology and future stability of the
  universe}, Oxford mathematical monographs, Oxford University Press, 2013.

\bibitem{SarbachZannias2014}
{\sc O.~Sarbach and T.~Zannias}, {\em {T}he geometry of the tangent bundle and
  the relativistic kinetic theory of gases}, Classical and Quantum Gravity, 31
  (2014), p.~085013.

\bibitem{Taylor2017}
{\sc M.~Taylor}, {\em {T}he {G}lobal {N}onlinear {S}tability of {M}inkowski
  {S}pace for the {M}assless {E}instein--{V}lasov {S}ystem}, Annals of PDE, 3
  (2017), p.~9.

\bibitem{WainwrightEllis1997}
{\sc J.~Wainwright and G.~F.~R. Ellis}, eds., {\em {D}ynamical systems in
  cosmology}, Cambridge University Press, 1997.

\bibitem{WainwrightEtAl1999}
{\sc J.~Wainwright, M.~J. Hancock, and C.~Uggla}, {\em {A}symptotic
  self-similarity breaking at late times in cosmology}, Classical and Quantum
  Gravity, 16 (1999), pp.~2577--2598.

\end{thebibliography}
\bibliographystyle{siam}

\end{document}